\documentclass[a4paper,USenglish,cleveref,numberwithinsect]{lipics-v2019}
\usepackage[x11names]{xcolor}
\title{Gathering on a Circle with Limited Visibility by Anonymous Oblivious Robots}

\titlerunning{Gathering on a Circle with Limited Visibility by Anonymous Oblivious Robots}

\author{Giuseppe A. Di Luna}{DIAG, Sapienza University of Rome, Italy}{diluna@diag.uniroma1.it}{}{}%TODO mandatory, please use full name; only 1 author per \author macro; first two parameters are mandatory, other parameters can be empty. Please provide at least the name of the affiliation and the country. The full address is optional

\author{Ryuhei Uehara}{School of Information Science, JAIST, Japan}{uehara@jaist.ac.jp}{}{}

\author{Giovanni Viglietta}{Department of Computer Science and Engineering, University of Aizu, Japan}{viglietta@gmail.com}{}{}

\author{Yukiko Yamauchi}{Department of Informatics, Graduate School of ISEE, Kyushu University, Japan}{yamauchi@inf.kyushu-u.ac.jp}{}{}

\authorrunning{G.\,A. Di Luna, R. Uehara, G. Viglietta, and Y. Yamauchi}

\Copyright{Giuseppe A. Di Luna, Ryuhei Uehara, Giovanni Viglietta, and Yukiko Yamauchi} %TODO mandatory, please use full first names. LIPIcs license is "CC-BY";  http://creativecommons.org/licenses/by/3.0/


\nolinenumbers %uncomment to disable line numbering

\hideLIPIcs  %uncomment to remove references to LIPIcs series (logo, DOI, ...), e.g. when preparing a pre-final version to be uploaded to arXiv or another public repository

\begin{document}

\maketitle

\begin{abstract}
A swarm of anonymous oblivious mobile robots, operating in deterministic Look-Compute-Move cycles, is confined within a circular track. All robots agree on the clockwise direction (chirality), they are activated by an adversarial semi-synchronous scheduler (SSYNCH), and an active robot always reaches the destination point it computes (rigidity). Robots have limited visibility: each robot can see only the points on the circle that have an angular distance strictly smaller than a constant $\vartheta$ from the robot's current location, where $0<\vartheta\leq\pi$ (angles are expressed in radians).

We study the Gathering problem for such a swarm of robots: that is, all robots are initially in distinct locations on the circle, and their task is to reach the same point on the circle in a finite number of turns, regardless of the way they are activated by the scheduler. Note that, due to the anonymity of the robots, this task is impossible if the initial configuration is rotationally symmetric; hence, we have to make the assumption that the initial configuration be rotationally asymmetric.

We prove that, if $\vartheta=\pi$ (i.e., each robot can see the entire circle except its antipodal point), there is a distributed algorithm that solves the Gathering problem for swarms of any size. By contrast, we also prove that, if $\vartheta\leq \pi/2$, no distributed algorithm solves the Gathering problem, regardless of the size of the swarm, even under the assumption that the initial configuration is rotationally asymmetric and the visibility graph of the robots is connected.

The latter impossibility result relies on a probabilistic technique based on random perturbations, which is novel in the context of anonymous mobile robots. Such a technique is of independent interest, and immediately applies to other Pattern-Formation problems.
\end{abstract}

\section{Introduction}\label{sec:1}
\paragraph*{Background}
One of the most popular models for distributed mobile robotics is the {\em Look-Compute-Move} (LCM)~\cite{xxFlPS12,book2}. In this model, a Euclidean space, usually the real plane $\mathbb{R}^2$, is inhabited by a ``swarm'' of punctiform and autonomous computational entities, the {\em robots}. Each robot, upon activation, takes a snapshot of the space (Look), uses this snapshot to compute its destination (Compute), and then reaches its destination point (Move).

The activation pattern of the robots is controlled by an external scheduler. At one end of the synchrony spectrum, there is the {\em fully-synchronous} scheduler (FSYNCH): in this case, time is divided into discrete units (the {\em turns}). At each turn, the entire swarm is activated, and all robots synchronously execute one LCM cycle. At the other end, there is the {\em asynchronous} scheduler (ASYNCH), where robot activations are independent, and LCM cycles are not synchronized. Somewhere halfway, there is the {\em semi-synchronous} scheduler (SSYNCH), which activates an arbitrary subset of robots at each turn, with the restriction of activating each robot infinitely often.

A common assumption in the context of mobile robots is the lack of persistent memory: a robot does not remember anything about past activations ({\em obliviousness}). Other assumptions are {\em anonymity}, where robots do not have visible and distinguishable identifying features, and {\em silence}, where robots do not have explicit communication primitives. This robot model is sometimes referred to as ``OBLOT''~\cite{book2}.

Obliviousness, anonymity, and silence are practical, useful, and desirable properties: an algorithm for oblivious robots is inherently resilient to transient memory failures; one for anonymous robots is ideal in privacy-sensitive contexts; an algorithm for silent robots works even in scenarios where communication is jammed or unfeasible (e.g., hostile environments or underwater deployment).

The purpose of such an ensemble of weak robots is to reach a common goal in a coordinated way. Interestingly, it has been shown that mobile robots can solve an extensive set of problems~\cite{book2}, ranging from forming patterns~\cite{xxFlPSV17,xxFuYKY15,xxSuzY99,xYuUKY17} to simulating a powerful Turing-complete movable entity~\cite{LunaFSV18}.

Among all tasks, a particularly relevant one is {\em Gathering}~\cite{xAgP06,xCiFPS12,xseb2,xgatheringdefago,Flocchini2019,xPaPV15,PoudelS17}: in finite time, all robots have to reach the same point and stop there. Initial works assumed robots to see the entire space ({\em full visibility}). However, a more realistic assumption~\cite{Bhagat2019,MondeYKY17} is that a robot be able to see only a portion of the space ({\em limited visibility}). 

In this paper, we study the Gathering problem for a swarm of oblivious robots with limited visibility constrained to move within a circle: each robot can see only the points on the circle that have an angular distance strictly smaller than a certain {\em visibility range} $\vartheta$. We assume that robots have no agreement on common coordinates apart from sharing the same notion of clockwise direction on the circle.

From a practical perspective, the restriction of moving along a predetermined path arises in wide variety of scenarios: railway lines, roads, tunnels, waterways, etc. We argue that the circle is the most meaningful curve to study: a solution for it readily extends to all other closed curves. 

From a theoretical perspective, confining the swarm on a circle (hence, a non-simply connected space) rules out all the strategies typically used for robots in the plane, such as moving toward the center of the visible set of robots (an example is in~\cite{xxAnOSY99}). Moreover, robots cannot use any asymmetries in the environment to identify a gathering point: this makes the circle the most challenging setting for Gathering (and in general, for any problem where symmetry breaking helps). 

Apart from~\cite{FLOCCHINI200867}, which examined the problem of scattering on a circle (reaching a final configuration where the robots are uniformly spaced out), no other works studied the computational power of oblivious robots when confined to curves: this is rather surprising, considering the copious existing literature on oblivious robots~\cite{xxFlPS12,book2}. To the best of our knowledge, the present paper is the first to investigate the Gathering problem for oblivious, silent, and anonymous robots on a circle with limited visibility.
 
\paragraph*{Our contributions}
We consider a swarm of $n$ oblivious, anonymous, and silent robots that start at distinct locations on a circle. Robots do not agree on a common system of coordinates, but they do share the same {\em handedness} (i.e., they have a common notion of clockwise direction). When a robot decides to move, it reaches its destination point (robots are {\em rigid}). Moreover, robots have no information on the swarm's size, $n$. Each robot can see only the points on the circle that have an angular distance strictly smaller than a certain visibility range $\vartheta$. We must assume that the initial configuration is rotationally asymmetric, otherwise the scheduler may activate robots in such a way as to preserve the rotational symmetry, and Gathering cannot be achieved.

After giving all the necessary definitions and some preliminary results (\cref{sec:2}),
our first contribution is to show that there is no distributed algorithm that solves the Gathering problem in SSYNCH when $\vartheta\leq \pi/2$, i.e., each robot is only able to see at most half of the circle (\cref{sec:3}). Surprisingly, this holds even if the initial configuration is rotationally asymmetric, the visibility graph of the swarm (i.e.,  the graph of intervisibility between robots) is connected, and all robots know $n$.

Our proof uses a novel technique based on random perturbations, of which we offer an intuitive probabilistic argument, as well as a formal and more elementary proof by derandomization. We show that, for any given distributed algorithm, either there exists an asymmetric configuration of robots that can evolve into a symmetric one within one time unit (in SSYNCH), or there is an asymmetric configuration where no robot can move. In either case, Gathering is impossible.

We stress that our result has a profound meaning, since it shows that, when $\vartheta\leq \pi/2$, any distributed algorithm, including the ones that do not aim to solve Gathering, has an initial asymmetric configuration that either repeats forever or evolves into a symmetric configuration in one step. This implies a novel impossibility result for geometric Pattern Formation on circles: even when robots start from an asymmetric configuration, they cannot form a target asymmetric pattern. This is in striking contrast with the unlimited-visibility setting, where, even under the ASYNCH scheduler, any pattern can be formed from any asymmetric configuration~\cite{book2}.

To the best of our knowledge, this the first impossibility proof for oblivious robots that neither relies on invariants induced by symmetries (e.g.,~\cite{xxFuYKY15,Yamauchi19}) nor on the disconnection of the visibility graph (e.g.,~\cite{LunaFSV18,YamauchiY13}). Due to the above, we think that our technique is of independent interest, and its core ideas could be applied to other settings, as well.

On the possibility side, we show that, if $\vartheta=\pi$ (i.e., each robot can see the entire circle except its antipodal point), there is a distributed algorithm that solves the Gathering problem in SSYNCH for swarms of any size (\cref{sec:4}). The algorithm's strategy is to attempt to elect a unique leader and form a multiplicity point, where all robots will subsequently gather. The main challenge is that, since a robot ignores whether its antipodal point is occupied by another robot or not (robots do not know $n$), there may be an ambiguity on who is the true leader. Several robots may believe to be the leader, but this also comes with the awareness of the possibility of being wrong: these ``undecided'' robots will make some adjustment moves, which eventually result in a configuration where one robot is absolutely certain of being the true leader. The leader will then form a multiplicity point by moving to another robot, and finally all other robots will join them.

A conference version of this paper has appeared at DISC 2020~\cite{DISC2020}.

\paragraph*{Related work}
The relevant literature can be divided into (i)~works that study the Gathering problem in the plane, (ii)~works where robots are confined on a circle (or a polygon with holes) but use memory, randomness, or other forms of symmetry breaking (e.g., different speeds), and (iii)~works that consider oblivious robots but a discrete space (i.e., a ring graph).

\noindent {\bf Gathering in the plane.} Several papers studied the Gathering problem for oblivious robots in the plane: for recent surveys, see~\cite{Flocchini2019,Bhagat2019}. In the following, we will focus only on the ones that limit the visibility of robots. In~\cite{xxAnOSY99}, a simple Gathering strategy is proposed for the FSYNCH setting, which consists in moving toward the center of the local smallest enclosing circle without losing vision of other robots (in SSYNCH, this only guarantees that all robots will converge to the same point). Interestingly, this strategy does not need chirality or any agreement on the robots' local coordinate systems. Recently, it was shown that this algorithm terminates in a number of rounds that is quadratic in the diameter of the visibility graph of the initial configuration of the robots~\cite{OPODIS22}.

In~\cite{xxFlPSW05}, Gathering is solved in ASYNCH, assuming that all robots agree on the orientation of both coordinate axes. This induces an implicit agreement on a total order among all robots, and such an order is at the base of the algorithm: each robot moves to its rightmost neighbor; if none is visible, it moves to its topmost neighbor. In~\cite{defago09}, the Gathering problem is solved in SSYNCH, assuming an eventual agreement on the North direction.

A recent paper, \cite{DEFECTVIEW}, investigated the assembly of robots with a ``defective view''. In this model, the snapshot taken by a robot captures a subset of up to $k\leq n$ robots. Two defective models are discussed: an adversarial one, where the subset is chosen by an adversary, and a distance-based one, where the snapshot includes the $k$ nearest robots. The paper presents solution algorithms for $k = n-2$ in the adversarial setting (when $n \geq 5$) and for $k = 2$ and $n = 4$ in the distance-based setting. Additionally, the paper demonstrates the impossibility of solving the problem when $k = 1$ and $n = 3$ in both adversarial and distance-based models.

All the aforementioned works heavily rely on the ability to freely move within a plane or on agreements on at least one direction.

\noindent{\bf Gathering using memory, randomness, or different speeds.} In~\cite{flocchininew}, the Gathering problem on a circle is studied in the continuous model, where robots observe and update their movement direction at every instant. The paper assumes that robots have persistent memory and access to a randomness source; it also assumes that robots can only see other robots and communicate with them when they are co-located. To compensate for this limited visibility, it is assumed that the robots either know the length of the circle or have an upper bound on the total number of robots.

The proposed algorithms employ a strategy wherein a robot is first elected as a leader, who then gathers the remaining robots. It is worth noting that the reliance on randomness (where robots can flip random coins) and persistent memory means that the techniques used in this work are not applicable to our setting.

Other works that investigated Gathering on a circle in the continuous model are~\cite{speed1,KranakisKMS15}. However, the solutions proposed in these papers strongly rely on either randomness or different robot speeds, both of which are ways to break symmetry.

Finally, there is a series of papers~\cite{pelc1,pelc2,pelc3} that solve a weaker version of Gathering, where robots are required to reach the same point, even if they are unaware of each other and do not stop at that point. In polygons with holes, however, these algorithms require persistent memory. An exception is~\cite{meetingpolygonvigiletta}, which studies a related problem for oblivious robots in a polygonal environment, and proposes a strategy to simulate persistent memory by carefully moving within a neighborhood of a polygon's vertex.

\noindent {\bf Gathering in the discrete ring.} A large body of works investigated Gathering in a ring graph~\cite{dangelo,dangelo2,DILUNA202079,gatheringlight,ringbook}, with the majority of papers assuming persistent memory~\cite{ringbook}, sometimes in the form of visible lights~\cite{gatheringlight}. When oblivious agents are considered, and under the assumption of full visibility, the problem has been entirely characterized~\cite{dangelo,dangelo2,KLASING200827}. Some works explored oblivious agents that can only see their immediate neighborhood: in~\cite{7000135}, a solution to the Gathering problem is given, assuming weak multiplicity detection and knowledge of the number of agents. On the other hand, the results in~\cite{GUILBAULT201386} imply that Gathering is impossible if the ring has an even number of nodes and the agents are not initially located on three consecutive nodes. We remark that the results for the discrete setting do not extend to our continuous circle.

\section{Model definition and preliminaries}\label{sec:2}
\paragraph*{Measuring angles}
Let $C\subset \mathbb R^2$ be a circle, and let $a$ and $b$ be points of $C$. The \emph{angular distance} between $a$ and $b$ (with respect to $C$) is the measure of the angle subtended at the center of $C$ by the shorter arc with endpoints $a$ and $b$. It follows that the angular distance between two points is a real number in the interval $[0,\pi]$,
where angles are expressed in radians. Two points of $C$ are \emph{antipodal} of each other if their angular distance is $\pi$. The \emph{$\alpha$-neighborhood} of a point $q\in C$ is the set of points of $C$ whose angular distance from $q$ is strictly smaller than $\alpha$. The $(\pi/2)$-neighborhood of $q$ is also called the \emph{open semicircle} centered at $q$.

Furthermore, if $a$ and $b$ are distinct points of $C$, we define $cw(a,b)$ as the measure of the \emph{clockwise} angle $\angle acb$, where $c$ is the center of $C$. Note that the order of the two arguments matters, and so for instance $cw(a,b)+cw(b,a)=2\pi$. We also define $cw(a,a)=0$ for every $a$.

\paragraph*{Rotational symmetry}
Let $S$ be a finite multiset of points on a circle $C$. We say that $S$ is \emph{rotationally symmetric} if there is a non-identical rotation around the center of $C$ that leaves $S$ unchanged (also preserving multiplicities). It is easy to see that the angle of rotation must be $2\pi/k$ radians for some integer $k>1$. If $S$ is not rotationally symmetric, it is said to be \emph{rotationally asymmetric}.

We now present a sufficient condition for a set of points to be rotationally asymmetric. This will be used repeatedly in \cref{sec:3}.

\begin{proposition}\label{p:antipodalasymmetry}
Let $S$ be a finite set of points on a circle. If there are exactly two points of $S$ whose antipodal points are not in $S$, then $S$ is rotationally asymmetric.
\end{proposition}
\begin{proof}
Let $p$ and $p'$ be two distinct points of $S$ whose antipodal points are not in $S$. Assume for the sake of contradiction that there is an integer $k>1$ such that $\rho(S)=S$, where $\rho$ is the rotation by $2\pi/k$ radians around the center of the circle. Note that two points $a$ and $b$ on the circle are antipodal to each other if and only if $\rho(a)$ and $\rho(b)$ are antipodal to each other. It follows that $\rho$ must map a point whose antipodal is not in $S$ into another point whose antipodal is not in $S$. Hence $\rho(p)=p'$ and $\rho(p')=p$, implying that $\rho\circ \rho$ is the identity map, and therefore $k=2$. But this means that $\rho$ maps every point to its antipodal, and so $p$ and $p'$ are antipodal to each other, which contradicts the fact that their antipodal points are not in $S$.
\end{proof}

\paragraph*{Angle sequences}
Let $S$ be a multiset of $n$ points on a circle $C$, and let $p\in S$. Let $p_1$, $p_2$, \dots, $p_n$ be the points of $S$ taken in clockwise order starting from $p=p_1$ (coincident elements of $S$ are ordered arbitrarily). We define the \emph{angle sequence} of $p$ (with respect to $S$) as the $n$-tuple $(cw(p_1,p_2), cw(p_2,p_3), \dots, cw(p_n,p_1))$. The case where all the elements of $S$ are coincident is an exception, and in this case the angle sequence of the $i$th point of $S$, with $1\leq i\leq n$, is defined as the $n$-tuple $(0,0,\dots,0,0,2\pi,0,0,\dots,0,0)$, where the term $2\pi$ appears in the $i$th position. Note that, with this convention, the sum of the elements of any angle sequence is always $2\pi$.

The next two propositions are fundamental results about angle sequences. Although they are easy observations, they will be used repeatedly in the rest of the paper.
\begin{proposition}\label{p:leader}
A non-empty multiset of points on a circle is rotationally asymmetric if and only if all its points have distinct angle sequences.\qed
\end{proposition}

With the above notation, let $q\in C$, and let $j$, with $1\leq j\leq n$, be the unique index such that $0<cw(p_j,q)\leq cw(p_j,p_{j+1})$, with the convention that $p_{n+1}=p_1$ (i.e., $q$ lies on the clockwise arc $p_jp_{j+1}$, and $q\neq p_j$). We say that the angle sequence of $p$ \emph{truncated at $q$} is the $j$-tuple $(cw(p_1,p_2), cw(p_2,p_3), \dots, cw(p_{j-1},p_j), cw(p_j,q))$.

If $W_1$ and $W_2$ are two (truncated) angle sequences, we write $W_1\prec W_2$ if $W_1$ is lexicographically smaller than $W_2$ (if $W_1$ and $W_2$ do not have the same length, we first pad the shorter sequence at the end with enough zeros). We write $W_1\preceq W_2$ to mean $W_1\prec W_2$ or $W_1=W_2$, and so on. Also, to denote the concatenation of two (truncated) angle sequences $W_1$ and $W_2$, we write $W_1W_2$.
\begin{proposition}\label{p:angle}
Let $S$ be a multiset of points on a circle $C$, let $p,p'\in S$, and let $W$ (respectively, $W'$) be the angle sequence of $p$ (respectively, $p'$). Let $q,q'\in C$ such that $cw(p,q)=cw(p',q')$, and let $Z$ (respectively, $Z'$) be the angle sequence $W$ (respectively, $W'$) truncated at $q$ (respectively, $q'$). If $W\preceq W'$, then $Z\preceq Z'$.\qed
\end{proposition}

\paragraph*{Mobile robots}
Our model of mobile robots is among the standard ones defined in~\cite{xxFlPS12,book2}. A swarm of $n>1$ robots is located on a circle $C\subset\mathbb R^2$, where each robot is a computational unit that occupies a point of $C$ (which may change over time) and operates in deterministic Look-Compute-Move cycles.

Time is discretized and subdivided into units, and at each time unit an adversarial (\emph{semi-synchronous}) \emph{scheduler} decides which robots are active and which are inactive. An inactive robot remains idle for that time unit, whereas an active robot takes a \emph{snapshot} of its surroundings, consisting of an arc $B\subseteq C$ and a list of points of $B$ that are currently occupied by robots, it computes a destination point in $B$ as a function of the snapshot, and it instantly moves to the destination point. The only restriction to the scheduler is that no robot should remain inactive for infinitely many consecutive time units.

Robots may have \emph{full visibility}, in which case the arc $B$ defining a snapshot coincides with the entire circle $C$, or they may have \emph{limited visibility}, in which case the arc $B$ consists of the $\vartheta$-neighborhood of the current position of the robot taking the snapshot, where $\vartheta$ is a positive constant called the \emph{visibility range} of the robots.

Furthermore, each robot has its own \emph{local coordinate system}, meaning that each snapshot it takes of an arc $B\subseteq C$ is actually a roto-translated copy of $B$ and the positions of the robots within $B$. Such a copy of $B$ has its midpoint at the origin of the coordinate system (this corresponds to the location of the robot taking the snapshot) and its endpoints have non-negative $x$ coordinate and the same $y$ coordinate.

Robots are also capable of \emph{weak multiplicity detection}, meaning that the snapshots they take contain some information on how many robots occupy each location. Specifically, a robot can tell if a point in a snapshot contains no robots, exactly one robot, or more than one robot: no information on the precise number of robots is given if this number is greater than one. A point occupied by more than one robot is also called a \emph{multiplicity point}.

In order to simplify our notation, when no confusion arises, we will often identify a robot with its position on the circle. So, we may improperly refer to a robot as a point $p\in C$ or to a swarm of robots as a set $S\subset C$.

\paragraph*{Gathering}
A \emph{distributed algorithm} is a function that maps a snapshot to a point within the snapshot itself. A robot \emph{executes} a distributed algorithm $A$ if, whenever it is activated and takes a snapshot $Q$, it moves to the destination point corresponding to $A(Q)$. In other words, at each time unit, an active robot chooses its destination point deterministically within its visibility range, based solely on the snapshot it currently has.

We stress that, as a consequence of the previous definitions, the robots in this model are \emph{oblivious} (i.e., they have no memory of past observations), \emph{anonymous} (i.e., a robot only identifies other robots by their positions in its local coordinate system, and not for instance by their IDs), \emph{silent} (i.e., they cannot send messages to one another), \emph{deterministic} (i.e., they cannot flip coins), \emph{rigid} (i.e., they always reach the destination points they compute), they have \emph{chirality} (i.e., they all agree on the clockwise direction on the circle), and they have no knowledge of $n$ (i.e., a robot can only see other robots within its visibility range, and it does not know whether there are further robots outside of it).

We say that a distributed algorithm $A$ solves the \emph{Gathering problem} under condition $P$ if, whenever all the $n>1$ robots of a swarm located on a circle execute $A$, they eventually reach a configuration where all robots are in the same point of the circle and no robot ever moves again, provided that their initial configuration satisfies condition $P$, and regardless of the activation choices of the adversarial scheduler. Equivalently, we say that $A$ is a \emph{Gathering algorithm} under condition $P$.

We remark that all the robots in the swarm must execute the same algorithm $A$ (i.e., robots are \emph{uniform}), and the algorithm has to work for swarms of any size $n>1$, where $n$ is \emph{not} a parameter of $A$. Also note that the robots' positions should not simply converge to the same limit, but they must actually become coincident in a finite number of time units for Gathering to be achieved (albeit there is no bound on the number of time units this process may take).

\paragraph*{Initial conditions}
There are several meaningful options concerning our choice of the initial condition $P$ for the Gathering problem. A typical assumption is that the $n$ robots be initially located in $n$ distinct points of the circle: while not strictly necessary, this is a common requirement for the Gathering problem (e.g.,~\cite{xxSuzY99,xCiFPS12,Flocchini2019}).

Another assumption that we may make is that the \emph{visibility graph} of the robots be initially connected. By ``visibility graph'' we mean the graph whose nodes are the $n$ robots, where there is an edge between two robots if and only if they are mutually visible, i.e., if their angular distance is less than $\vartheta$. This assumption is another common one (e.g.,~\cite{Bhagat2019,xxAnOSY99,xxFlPSW05,defago09,degener11}) and, although not strictly necessary, it is justified by the intuition that different connected components of the visibility graph may never become aware of each other, and therefore may fail to gather. We will make this assumption in \cref{sec:3} to strengthen our impossibility result, and we will not need to explicitly make it in \cref{sec:4}, because it will come as a consequence of other assumptions.

An important mandatory condition is that the multiset of the robots' positions on the circle should be rotationally asymmetric, due to the following.
\begin{proposition}\label{p:symmetry}
Let $S$ be any rotationally symmetric multiset of $n>1$ points on a circle. There is no Gathering algorithm under the initial condition that the multiset of robots' positions is $S$.
\end{proposition}
\begin{proof}
Let $S$ have a $k$-fold rotational symmetry, with $k>1$. This means that a rotation by $2\pi/k$ radians around the center of the circle leaves $S$ unchanged. So, the swarm can be partitioned into $n/k$ classes, where each class consists of $k$ robots located at the vertices of a regular $k$-gon.

It follows that, if the scheduler activates all robots in the swarm, then the $k$ robots in a same class get identical snapshots (regardless of the value of $\vartheta$), and therefore move to $k$ destination points that once again are the vertices of a regular $k$-gon. Hence, the new configuration has a $k'$-fold rotational symmetry, with $k'\geq k$.

So, if the scheduler keeps activating all robots at every time unit, no two robots in the same class will ever reach the same point, and the Gathering problem will never be solved.
\end{proof}
Since the robots are oblivious, this condition should hold true not only at the beginning, but at all times during the execution of a Gathering algorithm: the robots should never ``accidentally'' form a rotationally symmetric multiset, or they will be unable to gather.
\begin{corollary}\label{c:symmetry}
Throughout the execution of any Gathering algorithm, the robots' positions must always form a rotationally asymmetric multiset.\qed
\end{corollary}

\section{Impossibility of Gathering for $\vartheta\leq \pi/2$}\label{sec:3}
\paragraph*{Overview}
In this section we prove that, if each robot can see at most an open semicircle (i.e., $\vartheta\leq \pi/2$), then no distributed algorithm solves the Gathering problem, even under some strong assumptions on the initial configuration, and even if the robots know the size of the swarm.

Our technique is essentially probabilistic, and it starts by defining a set of perturbations of a regular configuration. Then, by analyzing the behavior of a generic distributed algorithm on all perturbations of a swarm that satisfy some initial conditions, we will show that the algorithm either (i)~allows the construction of a rotationally asymmetric configuration that can evolve into a rotationally symmetric one (under a semi-synchronous scheduler) or (ii)~leaves the configuration unchanged forever. In both cases, the algorithm does not solve the Gathering problem on some configurations.

\paragraph*{Summary}
We are interested in special types of perturbations of regular swarm configurations, depending on $\vartheta$. In \cref{p:size}, we show that the desired requirements of such perturbations can actually be satisfied by arbitrarily large swarms of robots. \cref{p:epsilon} demonstrates that these perturbations have good structural properties, such as preserving the visibility graph of swarm configurations (\cref{c:isomorphic}).

Leveraging the preceding statements, we provide strong constraints on how robots must behave under a Gathering algorithm. Essentially, we argue that at every step, all robots must move to locations already occupied by robots (\cref{l:imposs1}). Furthermore, if a Gathering algorithm prescribes that in a given configuration, a robot $r_i$ moves to another robot $r_j$, perturbing $r_i$ will never cause it to move to the same robot $r_j$ (\cref{l:imposs2}).

With these lemmas in place, it is relatively easy to conclude that there is always a perturbation of the regular configuration that causes any given algorithm to fail to make all robots gather. We formally prove this in \cref{t:impossible}, making use of the technical result of \cref{l:derandomize}.

\paragraph*{Perturbations}
For the rest of this section, we will denote by $C$ the unit circle centered at the origin. A finite set $S\subset C$ is \emph{regular} if $(1,0)\in S$ and all points of $S$ have the same angle sequence. Hence, for every positive integer $n$, there is a unique regular set of size $n$: for $n\geq 3$, this is the set of vertices of the regular $n$-gon centered at the origin and having a vertex in $(1,0)$.

Let $S$ be the regular set of size $n$, and let $p_1$, $p_2$, \dots, $p_n$ be the points of $S$ taken in clockwise order, starting from $p_1=(1,0)$. Let $\varepsilon\in \mathbb R$ with $0<\varepsilon<2\pi/n$, and let $\gamma=(\gamma_1, \gamma_2,\dots,\gamma_n) \in [0,1]^n\subset \mathbb R^n$. The \emph{$\varepsilon$-perturbation} of $S$ with \emph{coefficients} $\gamma$ is the set $S'\subset C$ of size $n$ such that, for all $1\leq i\leq n$, there is a (unique) point $p'_i\in S'$ with $cw(p_i,p'_i)=\gamma_i\cdot \varepsilon$, called the \emph{perturbed copy} of $p_i$. So, any $\varepsilon$-perturbation of $S$ is obtained by rotating each point of $S$ clockwise around the origin by an angle in $[0,\varepsilon]$.

Furthermore, for $1\leq i\leq n$, we say that two coefficient $n$-tuples $\gamma=(\gamma_1, \gamma_2, \dots, \gamma_n)$ and $\gamma'=(\gamma'_1, \gamma'_2, \dots, \gamma'_n)$ are \emph{$i$-related} if and only if they differ at most by their $i$th terms, i.e., for all $j\neq i$, we have $\gamma_j=\gamma'_j$. Note that the $i$-relation is an equivalence relation on $[0,1]^n$. With the previous paragraph's notation, we say that the set of all $\varepsilon$-perturbations of $S$ whose coefficients are in a same equivalence class of the $i$-relation is a \emph{bundle} of $\varepsilon$-perturbations of $p_i$. Intuitively, a bundle of $\varepsilon$-perturbations of $p_i$ is obtained by first fixing a perturbation of all points of $S$ except $p_i$, and then perturbing $p_i$ in all possible ways.

\paragraph*{Size of the swarm}
We will prove that Gathering is impossible for any given visibility range $\vartheta\leq \pi/2$, provided that the size of the swarm $n$ is appropriate. Specifically, we say that a positive integer $n$ is \emph{compatible} with $\vartheta$ if three conditions hold on the regular set $S$ of size $n$ (refer to \cref{f:compatible}):
\begin{enumerate}
\item For every $p\in S$, the open semicircle centered at $p$ contains exactly half of the points of $S$.
\item No two points of $S$ have an angular distance of exactly $\vartheta$.
\item There are two distinct points of $S$ whose angular distance is smaller than $\vartheta$.
\end{enumerate}

We can show that there are arbitrarily large such integers:
\begin{proposition}\label{p:size}
For any $\vartheta\leq \pi/2$, there are arbitrarily large integers compatible with $\vartheta$.
\end{proposition}
\begin{proof}
It is easy to see that the first condition is satisfied if and only if $n$ is of the form $4k+2$, for some non-negative integer $k$.

Let us turn to the second condition. The angular distance between two points of $S$ is a number of the form $a\cdot 2\pi/n$, for some positive integer $a\leq n/2$. So, the second condition is satisfied if and only if there is no $a$ such that $a/n=\vartheta/2\pi$.

We will prove that, given any two integers of the form $4k+2$ and $4k+6$, i.e., two consecutive numbers satisfying the first condition, at least one of them satisfies the second condition: this will allow us to conclude that there are arbitrarily large integers that satisfy both conditions (\cref{f:compatible} shows an example for $k=3$).

For the sake of contradiction, let us assume the opposite: there exists $1\leq a\leq 2k+1$ such that $a/(4k+2) = \vartheta/2\pi$, and there exists $1\leq b\leq 2k+3$ such that $b/(4k+6) = \vartheta/2\pi$. By combining the two equations, we have $a/(4k+2)=b/(4k+6)$, or equivalently $a(2k+3)=b(2k+1)$. The numbers $2k+1$ and $2k+3$ are two consecutive odd integers, and therefore they are relatively prime. It follows that $2k+1$ must be a divisor of $a$, and so $a=2k+1$. This implies that $\vartheta=2\pi\cdot (2k+1)/(4k+2)=\pi$, contradicting the assumption that $\vartheta\leq \pi/2$.

To satisfy the third condition, simply pick $n>2\pi/\vartheta$.
\end{proof}

\begin{figure}%[h!]
\begin{center}
\includegraphics[scale=1]{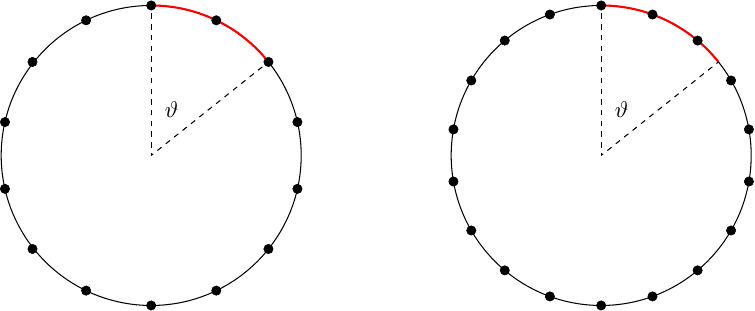}
\end{center}
\caption{(Left) $n=14$ is not compatible with $\vartheta=2\pi/7$ because it does not satisfy condition~2. (Right) $n=18$ is compatible with $\vartheta=2\pi/7$.}
\label{f:compatible}
\end{figure}

\paragraph*{Choice of $\varepsilon$}
For every integer $n$ compatible with $\vartheta$, we define a positive number $\varepsilon_{\vartheta, n}$, which will be used to construct perturbations of the regular set $S$ of size $n$. We set $\varepsilon_{\vartheta, n}=\delta/2$, where $\delta=\min\{|\vartheta - 2\pi a/n| \mid a\in\mathbb N,\ 0\leq a\leq n\}$.

Since $n$ is compatible with $\vartheta$, it easily follows that $\varepsilon_{\vartheta, n}>0$. Also, $\delta$ is at most half the angular distance between two consecutive points of $S$, and therefore $\varepsilon_{\vartheta, n}\leq \pi/2n$. Moreover, our choice of $\varepsilon_{\vartheta, n}$ has some other desirable properties:
\begin{proposition}\label{p:epsilon}
Let $n$ be an integer compatible with $\vartheta\leq \pi/2$, let $S$ be the regular set of size $n$, and let $S'$ be an $\varepsilon_{\vartheta, n}$-perturbation of $S$. If $p\in S$, and $p'\in S'$ is the perturbed copy of $p$, the following hold:
\begin{enumerate}
\item The $\vartheta$-neighborhood of $p$ contains a point $q\in S$ if and only if the $\vartheta$-neighborhood of $p'$ contains the perturbed copy of $q$ in $S'$.
\item The open semicircle $D$ centered at $p$ contains exactly half of the points of $S'$, which are the perturbed copies of the points of $S$ contained in $D$.
\item If $D'$ is the open semicircle centered at $p'$, then $S'\cap D=S'\cap D'$, and hence $D'$ contains exactly half of the points of $S'$.
\end{enumerate}
\end{proposition}
\begin{proof}
To prove the first statement, let $B$ and $B'$ be the $\vartheta$-neighborhoods of $p$ and $p'$, respectively. Let $a$ and $b$ be the endpoints of $B$, such that $cw(a,b)=2\vartheta$, and let $a'$ and $b'$ be the endpoints of $B'$, such that $cw(a',b')=2\vartheta$. We have $cw(a,a')=cw(b,b')=cw(p,p')\leq \varepsilon_{\vartheta, n}$. Also, by definition of $\varepsilon_{\vartheta, n}$, for every point $q\in S$, we have $cw(q,a)\geq 2\varepsilon_{\vartheta, n}$, $cw(a,q)\geq 2\varepsilon_{\vartheta, n}$, $cw(q,b)\geq 2\varepsilon_{\vartheta, n}$, and $cw(b,q)\geq 2\varepsilon_{\vartheta, n}$. Moreover, if $q'\in S'$ is the perturbed copy of $q$, we have $cw(q,q')\leq \varepsilon_{\vartheta, n}$.

Suppose now that $q\in B$, and so $cw(a,q)\in [2\varepsilon_{\vartheta, n},\  2\vartheta-2\varepsilon_{\vartheta, n}]$. This means that $cw(a',q)\in [\varepsilon_{\vartheta, n},\  2\vartheta-2\varepsilon_{\vartheta, n}]$, and also that $cw(a',q')\in [\varepsilon_{\vartheta, n},\  2\vartheta-\varepsilon_{\vartheta, n}]$, which in turn implies that $q'\in B'$. Conversely, if $q\notin B$, we have $cw(b,q)\in [2\varepsilon_{\vartheta, n},\  2\pi-2\vartheta-2\varepsilon_{\vartheta, n}]$, which leads to $cw(b',q)\in [\varepsilon_{\vartheta, n},\  2\pi-2\vartheta-2\varepsilon_{\vartheta, n}]$, and to $cw(b',q')\in [\varepsilon_{\vartheta, n},\  2\pi-2\vartheta-\varepsilon_{\vartheta, n}]$, implying that $q'\notin B'$.

The second and third statements are proved in a similar way. Note that, since $n$ is compatible with $\vartheta$, $D$ contains exactly $n/2$ points of $S$. Recall from the proof of \cref{p:size} that $n$ is of the form $4k+2$, and therefore the smallest angular distance between a point of $S$ and an endpoint of $D$ is $\pi/n\geq 2\varepsilon_{\vartheta, n}$. Now we can repeat the proof of the first statement verbatim, with $D$ and $D'$ instead of $B$ and $B'$, to conclude that a point $q\in S'$ lies in $D$ if and only if it lies in $D'$, which is true if and only if $q$ is the perturbed copy of a point of $S$ that lies in $D$. Since $|S'|=|S|=n$, both $D$ and $D'$ contain exactly $n/2$ points of $S'$.
\end{proof}
\begin{corollary}\label{c:isomorphic}
Let $n$ be an integer compatible with $\vartheta\leq \pi/2$, and let $S'$ be an $\varepsilon_{\vartheta, n}$-perturbation of the regular set $S$ of size $n$. If two swarms of robots form $S$ and $S'$ respectively, their visibility graphs are isomorphic.
\end{corollary}
\begin{proof}
The isomorphism is realized by mapping a robot in the first swarm, located in $p\in S$, to the robot in the second swarm located in the perturbed copy of $p$. By the first statement of \cref{p:epsilon}, there is an edge between two robots in the first swarm if and only if there is an edge between the corresponding robots in the second swarm.
\end{proof}

\paragraph*{Combining configurations}
Next we will describe a way of combining two configurations of robots into a new one that takes an open semicircle from each. This operation will be used to construct configurations of robots where a given distributed algorithm fails to make the robots gather.

Let $S_1$ and $S_2$ be two subsets of $C$, and let $D$ be an open semicircle. The \emph{$D$-combination} of $S_1$ and $S_2$ is defined as the set $(S_1\cap D)\cup\rho(S_2\cap D)$, where $\rho$ is the rotation of $\pi$ around the origin. In other words, this operation takes $S_1$, discards the points that do not lie in $D$, and replaces them with the points of $S_2$ that lie in $D$, by mapping them to their antipodal points.

\paragraph*{Preliminary lemmas}
We are now ready to give our first two lemmas, which deal with swarms forming perturbations of a regular configuration, and analyze the distributed algorithms that make robots move in some ways. The first lemma states that, if an algorithm makes a robot move to a point not currently occupied by another robot, then the algorithm cannot solve the Gathering problem.
\begin{lemma}\label{l:imposs1}
Let $A$ be a distributed algorithm, let $n$ be compatible with $\vartheta\leq \pi/2$, and consider a swarm of robots that forms an $\varepsilon_{\vartheta, n}$-perturbation $S'$ of the regular set of size $n$. If there is a robot that, executing $A$, moves to a point not in $S'$, then $A$ does not solve the Gathering problem, even under the condition that the swarm initially forms a rotationally asymmetric set of $n$ distinct points with a connected visibility graph.
\end{lemma}
\begin{proof}
Let $p'\in S'$, and let $D$ be the open semicircle centered at $p'$. By the third statement of \cref{p:epsilon}, $|S'\cap D|=n/2$. Assume that, if the robot located in $p'\in S'$ executes $A$, it moves to a point $q\notin S'$ (see \cref{f:lemmas12} (left)). Let $S''=(S'\setminus \{p'\})\cup\{q\}$, and let $Q$ be the $D$-combination of $S'$ and $S''$. Observe that, since $\vartheta\leq \pi/2$, the $\vartheta$-neighborhood of $p'$ is a subset of $D$, and therefore $q\in D$. Hence, $|S''\cap D|=|S'\cap D|=n/2$, implying that $|Q|=n$.

Let $q'\in Q$ be the antipodal point of $q$. Note that the only points of $Q$ whose antipodal points are not in $Q$ are $p'$ and $q'$. So, $Q$ satisfies the hypotheses of \cref{p:antipodalasymmetry}, and therefore it is rotationally asymmetric.

Next we will prove that $Q$ has a connected visibility graph. Let $p''$ be the antipodal point of $p'$, let $Q'=(Q\setminus \{q'\})\cup\{p''\}$, and observe that $Q'$ is an $\varepsilon_{\vartheta, n}$-perturbation of the regular set $S$ of size $n$. Since $n$ is compatible with $\vartheta$, the $\vartheta$-neighborhood of any point of $S$ contains at least two more points of $S$, one on each side of it. It follows that the visibility graph of $S$ has a Hamiltonian cycle, and so does the visibility graph of $Q'$, by \cref{c:isomorphic}. Hence, the visibility graph of $Q'\setminus\{p''\}$ has a Hamiltonian path. By the first statement of \cref{p:epsilon}, the $\vartheta$-neighborhood of any point of $Q'$ contains two more points of $Q'$, one on each side of it. Therefore, each point of the circle $C$ lies in the $\vartheta$-neighborhood of at least two points of $Q'$. So, each point of $C$ lies in the $\vartheta$-neighborhood of at least one point of $Q'\setminus\{p''\}$. In particular, this is true of $q'$. We conclude that the visibility graph of $(Q'\setminus\{p''\})\cup\{q'\}=Q$ contains a path $P$ connecting $n-1$ nodes (corresponding to $Q'\setminus\{p''\}$), plus a node (corresponding to $q'$) connected to some node of $P$, implying that the whole graph is connected.

Now, consider a swarm initially forming $Q$, which, as we recall, is a rotationally asymmetric set of $n$ distinct points with a connected visibility graph. Suppose that, in the first time unit, the scheduler activates only the robot in $p'$, which executes $A$. Since the $\vartheta$-neighborhood of $p'$ is contained in $D$, the robot gets the same snapshot as it would if the swarm's configuration were $S'$. Therefore the robot moves to $q$, and the swarm's new configuration is $R=(Q\setminus \{p'\})\cup\{q\}$. Observe that the antipodal of each point of $R$ is also in $R$, and so any two antipodal points of $R$ have the same angle sequence. Due to \cref{p:leader}, $R$ is rotationally symmetric, and hence $A$ is not a Gathering algorithm, by \cref{c:symmetry} (recall from \cref{sec:2} that, by definition, the existence of a single schedule that prevents the robots from gathering is sufficient to conclude that $A$ is not a Gathering algorithm).
\end{proof}

\begin{figure}%[h!]
\begin{center}
\includegraphics[scale=1]{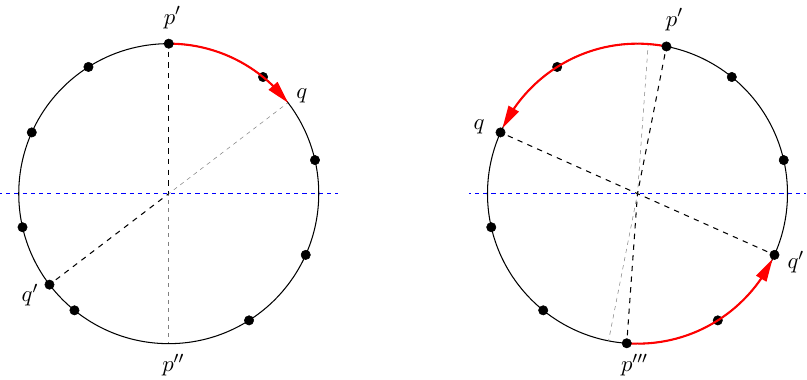}
\end{center}
\caption{Illustrations of \cref{l:imposs1} (left) and \cref{l:imposs2} (right).}
\label{f:lemmas12}
\end{figure}

The second lemma states that, if a distributed algorithm makes a robot $r$ move on top of another robot $r'$, and there is a perturbation of $r$ such that the same algorithm makes $r$ move on top of the same robot $r'$, then the algorithm does not solve the Gathering problem.
\begin{lemma}\label{l:imposs2}
Let $A$ be a distributed algorithm, let $n$ be compatible with $\vartheta\leq \pi/2$, let $S$ be the regular set of size $n$, and let $S'$ and $S''$ be two distinct sets in the same bundle of $\varepsilon_{\vartheta, n}$-perturbations of $p\in S$, where $p'\in S'$ and $p''\in S''$ are the perturbed copies of $p$. Assume that, if a swarm of robots forms $S'$ and the robot in $p'$ executes $A$, it moves to another robot, located in $q\in S'$. Also assume that, if a swarm of robots forms $S''$ and the robot in $p''$ executes $A$, it moves to the same point $q$. Then, $A$ does not solve the Gathering problem, even under the condition that the swarm initially forms a rotationally asymmetric set of $n$ distinct points with a connected visibility graph.
\end{lemma}
\begin{proof}
Let $D$ be the open semicircle centered at $p$, and let $Q$ be the $D$-combination of $S'$ and $S''$. By the second statement of \cref{p:epsilon}, $|S'\cap D|=|S''\cap D|=n/2$, and so $|Q|=n$.

Since $n$ is compatible with $\varepsilon$, the visibility graph of $S$ is connected. Observe that $Q$ is an $\varepsilon_{\vartheta, n}$-perturbation of $S$, and so, by \cref{c:isomorphic}, its visibility graph is isomorphic to that of $S$. This implies that the visibility graph of $Q$ is connected, as well.

Let $p'''\in Q$ be the antipodal point of $p''$. Since $p'\neq p''$, it follows that $p'$ and $p'''$ are not antipodal to each other. Moreover, as $S'$ and $S''$ are in the same bundle of $\varepsilon_{\vartheta, n}$-perturbations of $p$, we have $S'\setminus \{p'\}=S''\setminus \{p''\}$. As a consequence, the only points of $Q$ whose antipodal points are not in $Q$ are $p'$ and $p'''$. So, $Q$ satisfies the hypotheses of \cref{p:antipodalasymmetry}, and therefore it is rotationally asymmetric.

Consider a swarm initially forming $Q$, which, as we proved, is a rotationally asymmetric set of $n$ distinct points with a connected visibility graph. Suppose that, in the first time unit, the scheduler activates only the robots in $p'$ and $p'''$, which execute $A$. 

By the third statement of \cref{p:epsilon} applied to $S'$ and $Q$ (both of which are $\varepsilon_{\vartheta, n}$-perturbations of $S$), and since the $\vartheta$-neighborhood of $p'$ is contained in $D$, the robot in $p'$ gets the same snapshot as it would if the swarm's configuration were $S'$. Hence, this robot moves to $q$. Similarly, the robot in $p'''$ gets the same snapshot that it would if it were in $p''$ and the swarm's configuration were $S''$. So, this robot moves to $q'$, the antipodal point of $q$. Since $q$ and $q'$ are antipodal to each other, the swarm's new configuration is a rotationally symmetric multiset, because each robot has an antipodal one with the same angle sequence, and therefore \cref{p:leader} applies. We conclude that $A$ is not a Gathering algorithm, due to \cref{c:symmetry}.
\end{proof}

\paragraph*{Probabilistic argument}
Our concluding argument goes as follows. Suppose for a contradiction that there is a Gathering algorithm $A$ for some $\vartheta\leq \pi/2$. Let $n$ be an arbitrarily large integer compatible with $\vartheta$, and let $S$ be the regular set of size $n$. We will derive a contradiction by studying the behavior of $A$ on the swarms forming the $\varepsilon_{\vartheta, n}$-perturbations of $S$. Specifically, let $p_1$, $p_2$, \dots, $p_n$ be the points of $S$ taken in clockwise order, starting from $p_1=(1,0)$. Suppose that a swarm of $n$ robots forms a generic $\varepsilon_{\vartheta, n}$-perturbations of $S$, with robot $r_i$ occupying the perturbed copy of $p_i$, and let all robots execute algorithm $A$.

Let us first restrict ourselves to a bundle $P$ of $\varepsilon_{\vartheta, n}$-perturbations of some $p_i\in S$, and let us analyze the possible behaviors of the robot $r_i$. Recall that, by definition of bundle, the perturbations in $P$ have fixed coefficients for all the points of $S$ except $p_i$, and perturb $p_i$ in every possible way by varying the coefficient $\gamma_i\in [0,1]$. Observe that, by \cref{l:imposs1}, $A$ should never make $r_i$ move to some unoccupied location, or $A$ would not be a Gathering algorithm. Also, if two or more perturbations in the bundle $P$ made $r_i$ move to the same robot, then $A$ would not be a Gathering algorithm, due to \cref{l:imposs2}. However, by the pigeonhole principle, if $n$ perturbations in $P$ made $r_i$ move to some other robot, then at least two of them would make it move to the same robot. It follows that at most $n-1$ perturbations in $P$ can make $r_i$ move at all. So, all perturbations in $P$ except a finite number of them must make $r_i$ stay still.

Now, let us pick an $\varepsilon_{\vartheta, n}$-perturbation of $S$ by choosing its coefficients $\gamma\in [0,1]^n$ uniformly at random. Let us also define $n$ random variables $X_i\colon [0,1]^n\to \{0,1\}$, with $1\leq i\leq n$, such that $X_i(\gamma)=0$ if and only if algorithm $A$ makes the robot $r_i$ stay still when the swarm's configuration is the $\varepsilon_{\vartheta, n}$-perturbation of $S$ defined by the coefficients $\gamma$. By the above argument, for every bundle $P$ of $\varepsilon_{\vartheta, n}$-perturbations of $p_i$, we have $\Pr[X_i(\gamma)=1\mid \gamma\in P]=0$. Thus, with the convention that $\gamma=(\gamma_1, \gamma_2, \dots, \gamma_n)$, we have
\begin{equation*}
\begin{split}
\Pr[X_i=1] & = \int_{[0,1]^n}X_i(\gamma)\, d\gamma = \int_0^1\cdots\int_0^1 X_i(\gamma)\, d\gamma_1\, d\gamma_2 \cdots d\gamma_n \\
 & = \int_0^1\cdots\int_0^1 X_i(\gamma)\, d\gamma_i\, d\gamma_1 \cdots d\gamma_{i-1}\, d\gamma_{i+1} \cdots d\gamma_n \\
 & = \int_0^1\cdots\int_0^1 0\, d\gamma_1 \cdots d\gamma_{i-1}\, d\gamma_{i+1} \cdots d\gamma_n = 0.
\end{split}
\end{equation*}

Hence, the probability that $A$ will make the robot $r_i$ stay still when the swarm's configuration is picked at random among all $\varepsilon_{\vartheta, n}$-perturbations of $S$ is $1$. Since this is true of all robots separately, it is also true of all robots collectively, by the inclusion-exclusion principle. In other words, with probability $1$, on a random $\varepsilon_{\vartheta, n}$-perturbation of $S$, no robot will be able to move, and therefore the robots will be unable to gather. Moreover, with probability $1$, a random $\varepsilon_{\vartheta, n}$-perturbation of $S$ is rotationally asymmetric. As a consequence, there is at least one initial configuration (actually, a great deal of configurations) where the swarm forms a rotationally asymmetric set of $n$ distinct points with a connected visibility graph, and where no robot is able to move. We conclude that $A$ cannot be a Gathering algorithm, even under such strong conditions.

\paragraph*{Technical obstacles}
The probabilistic proof we outlined above is sound for the most part, but unfortunately making it rigorous is a delicate matter. The problem is that, in order for $X_i$ to be a random variable, it has to be a measurable function. For this to be true, the set of coefficients corresponding to perturbations where algorithm $A$ makes the robot $r_i$ stay still should be a measurable subset of $[0,1]^n$. In turn, this requires some assumptions on the nature of $A$, whereas we only defined $A$ as a generic function mapping a snapshot to a point.

However, since the function $A$ actually implements an algorithm, which typically is a finite sequence of operations that are well-behaved in an analytic sense, most reasonable assumptions on $A$ would rule out the pathological non-measurable cases, and would therefore make $X_i$ a properly defined random variable, allowing the rest of the proof to go through.

Nonetheless, we choose to adopt a different approach, which is both less technical and more general in scope. Indeed, we will give a ``derandomized'' version of the above proof, which will not deal with probability spaces and random variables, and will not require a more restrictive re-definition of which functions are computable by mobile robots.

\paragraph*{Derandomization}
Next we will show how to complete the previous argument without the use of probability. Note that we do not need to prove that a random perturbation causes all robots to stay still with probability $1$: we merely have to show that there is at least one perturbation with such a property. This is significantly easier, and is achieved by the next lemma, where $X_i$ no longer denotes a random variable but simply a set of coefficients.
\begin{lemma}\label{l:derandomize}
Let $m,n\in \mathbb N^+$, and let $X_1$, $X_2$, \dots, $X_n$ be subsets of the unit hypercube $[0,1]^n\subset \mathbb R^n$ such that every line parallel to the $i$th coordinate axis intersects $X_i$ in less than $m$ points, for all $1\leq i\leq n$. Then, there is a point in $[0,1]^n$ whose $n$ coordinates are all distinct that does not lie in any of the sets $X_1$, $X_2$, \dots, $X_n$.
\end{lemma}
\begin{proof}
For $1\leq i\leq n$, let $a_i=m^{2^{i-1}}$ and let $s_i=\sum_{j\leq i} a_j$, with $s_0=0$. Now let $Y_i=\{j/s_n \mid s_{i-1}< j\leq s_i\}\subset [0,1]$, and let $Z_i=Y_1\times Y_2\times\dots\times Y_i\subset [0,1]^i$. Observe that the sets $Y_i$ are mutually disjoint, and therefore all the points of $Z_n$ have distinct coordinates. We claim that $Z_n$ contains a point that does not lie in any of the sets $X_i$.

We will prove by induction on $i$ that, for all $x\in [0,1]^{n-i}$, the set $Z_i\times \{x\}\subset [0,1]^n$ contains a point that does not lie in $\bigcup_{j\leq i}X_j$. Thus, we will obtain our claim for $i=n$.

To prove the base step, consider the set $Z_1\times\{x\}$, with $x\in [0,1]^{n-1}$. This is a set of $|Z_1|=a_1=m$ points lying on a line parallel to the first coordinate axis, and therefore it must contain a point that does not lie in $X_1$.

For the inductive step, let $i>1$, and consider the set $Z_i\times\{x\}$, with $x\in [0,1]^{n-i}$. Since $Z_i=Z_{i-1}\times Y_i$, we have $Z_i\times\{x\} = \bigcup_{y\in Y_i} Z_{i-1}\times \{y\}\times \{x\}$. By the induction hypothesis, each of the $|Y_i|$ sets of the form $Z_{i-1}\times \{y\}\times \{x\}$ with $y\in Y_i$ contains a point not in $\bigcup_{j\leq i-1}X_j$. In other words, for each $y\in Y_i$ there exists $z\in Z_{i-1}$ such that the point $\{z\}\times \{y\}\times \{x\}$ is not in $\bigcup_{j\leq i-1}X_j$. Note that $|Z_{i-1}|=\prod_{j\leq i-1}a_j=\prod_{j\leq i-1}m^{2^{j-1}} = m^{2^{i-1}-1}=a_i/m = |Y_i|/m$, and hence $|Y_i|/|Z_{i-1}|=m$. So, by the pigeonhole principle, there exists $z\in Z_{i-1}$ such that there are at least $m$ different choices of $y\in Y_i$ for which the point $\{z\}\times \{y\}\times \{x\}$ is not in $\bigcup_{j\leq i-1}X_j$. Since these are at least $m$ points lying on a line parallel to the $i$th coordinate axis, at least one of them does not lie in $X_i$. It follows that such a point is in $Z_i\times\{x\}$ and does not lie in $\bigcup_{j\leq i}X_j$.
\end{proof}

We can now prove the main result of this section.
\begin{theorem}\label{t:impossible}
If $\vartheta\leq\pi/2$, and for arbitrarily large $n$, there is no Gathering algorithm under the condition that the swarm initially forms a rotationally asymmetric set of $n$ distinct points with a connected visibility graph.
\end{theorem}
\begin{proof}
Let $n$ be an arbitrarily large integer compatible with $\vartheta$, which exists due to \cref{p:size}. Note that all $\varepsilon_{\vartheta, n}$-perturbations of the regular set of size $S$ have a connected visibility graph, by \cref{c:isomorphic}. As before, we assume for a contradiction that $A$ is a Gathering algorithm, and we consider a swarm of size $n$ where all robots execute $A$, and each robot $r_i$ is initially located in the perturbed copy of point $p_i\in S$, for some $\varepsilon_{\vartheta, n}$-perturbation of $S$.

For $1\leq i\leq n$, let $X_i\subseteq [0,1]^n$ be the set of coefficients corresponding to perturbations where algorithm $A$ causes $r_i$ to make a non-null movement. As we already proved, due to \cref{l:imposs1,l:imposs2}, in each bundle of $\varepsilon_{\vartheta, n}$-perturbations of $p_i$, at most $n-1$ perturbations cause $r_i$ to move. Rephrased in geometric terms, every line in $\mathbb R^n$ parallel to the $i$th coordinate axis intersects $X_i$ in less than $n$ points.

So, the sets $X_1$, $X_2$, \dots, $X_n$ satisfy the hypotheses of \cref{l:derandomize} with $m=n$. As a consequence, there exists $\gamma=(\gamma_1,\gamma_2,\dots,\gamma_n)\in [0,1]^n$, where the coefficients $\gamma_i$ are all distinct, such that, in the perturbation corresponding to $\gamma$, algorithm $A$ causes all robots to stay still, and therefore does not allow them to gather.

It remains to check that the perturbation $S'$ corresponding to $\gamma$ is rotationally asymmetric. Let $p_1$, $p_2$, \dots, $p_n$ be the points of $S$ in clockwise order, and let $p'_i\in S'$ be the perturbed copy of $p_i$, for $1\leq i\leq n$. Suppose for a contradiction that $S'$ has a $k$-fold rotational symmetry with $k>1$, implying that the angular distance between $p'_1$ and $p'_{n/k+1}$ is $\alpha=2\pi/k$. Note that $\alpha$ is also the angular distance between $p_1$ and $p_{n/k+1}$. Moreover, by definition of perturbation, $\alpha=cw(p_1,p_{n/k+1})-cw(p_1,p'_1)+cw(p_{n/k+1},p'_{n/k+1})=2\pi/k-\gamma_1\cdot\varepsilon_{\vartheta, n}+\gamma_{n/k+1}\cdot\varepsilon_{\vartheta, n}$. It follows that $\gamma_1=\gamma_{n/k+1}$, which contradicts the fact that the coefficients $\gamma_i$ are all distinct (indeed, $k\geq 2$ implies that $1<n/k+1\leq n$).
\end{proof}

We remark that, throughout the proofs of \cref{l:imposs1,l:imposs2,t:impossible}, only swarms of the same size $n$ appear. It follows that our impossibility result holds even if all robots know the size of the swarm.

\section{Gathering algorithm for $\vartheta=\pi$}\label{sec:4}
\paragraph*{Overview}
In this section we give a Gathering algorithm for robots that can see the entire circle except their antipodal point (i.e., $\vartheta=\pi$), under the condition that the initial configuration is a rotationally asymmetric set with no multiplicity points.

First we will describe a simple Gathering algorithm for robots with full visibility, which already provides some useful ideas: elect a leader, form a unique multiplicity point, and gather there. We will then extend the same ideas to the limited-visibility case with $\vartheta=\pi$. There are some difficulties arising from the fact that not all robots will necessarily agree on the same leader, because they may have different views of the rest of the swarm. For instance, two antipodal robots will not see each other, and, if the configuration is rotationally asymmetric, they will obtain two non-isometric snapshots, which may cause them to elect two different leaders.

We will show how to cope with these difficulties. Essentially, based on what a robot $r$ knows, there are only two possibilities on who the ``true'' leader may be, depending on whether there is a robot antipodal to $r$ or not. If $r$ happens to be elected leader in both scenarios, then $r$ has no doubt of being the leader, and therefore behaves as in the full-visibility algorithm, creating a multiplicity point. In most cases, however, no robot will be so fortunate, but the swarm will still have to make some sort of progress toward gathering. So, the robots that see themselves as possible leaders (but could be wrong) make some preparatory moves that will ideally ``strengthen their leadership'' in the next turns. We will argue that, after a finite number of turns, one robot will become aware of being the leader and will create a multiplicity point, even under a semi-synchronous scheduler.

The design and analysis of our Gathering algorithm are further complicated by some undesirable special cases, where two distinct multiplicity points end up being created, or the multiplicity point is antipodal to some robot, and therefore invisible to it.

\paragraph*{Summary}
This section is structured as follows. We begin by describing a simple algorithm for the full-visibility case, then proceed to the scenario where $\vartheta=\pi$. In this context, we distinguish between the \emph{cognizant leader}, i.e., a robot certain of its leadership, and the \emph{undecided leaders}, i.e., robots that perceive themselves as leaders but are unsure due to incomplete information. \cref{p:possleader} states that the true leader must belong to one of these two categories.

Next, we establish a technical result, \cref{p:middle}, which supports the proof of the so-called \emph{point-addition lemma} (\cref{l:addition}). This crucial lemma describes a relationship between configurations that differ by only one robot; it leads to \cref{c:cognizant}, which provides a non-trivial characterization of a cognizant leader.

These preliminary results pave the way for our Gathering algorithm, which consists of several \emph{numbered rules}; any active robot applies a rule based on the configuration it observes. The rest of the section is dedicated to proving the correctness of this algorithm.

\cref{l:algomove} establishes several constraints on which robots can execute certain rules and under what configurations. A notable implication of \cref{l:algomove} is that only an undecided or cognizant leader may move (\cref{c:algomove}). Furthermore, \cref{l:algotwo} shows that at most two robots can move simultaneously, one of which is the true leader of the swarm.

These insights greatly simplify the analysis of potential executions of our Gathering algorithm. To begin with, it is relatively straightforward to prove that, if a multiplicity point is ever formed, the Gathering problem is eventually solved (\cref{l:algomult}). For all other scenarios, \cref{l:algocollision} shows that certain undesirable configurations do not occur as a result of executing the algorithm. Applying this knowledge, we establish that if a robot executes rule~3 or rule~4.a, the Gathering problem is eventually solved (\cref{l:algorule3,l:algorule4a}).

The rest of the analysis focuses on the remaining cases, where the only rules that are ever executed are rule~4.b and rule~4.c. \cref{l:algonoanti} shows that under these rules, the configuration remains rotationally asymmetric without multiplicity points; \cref{l:algoonemove} addresses the case where only the leader robot moves. Finally, \cref{t:algorithm} puts together all previous findings and concludes the proof of correctness.

\paragraph*{Full visibility and leader election}
We will describe a simple Gathering algorithm for the scenario where robots have full visibility.

Let $S$ be a rotationally asymmetric finite set of points on a circle. Recall from \cref{p:leader} that all points of $S$ have distinct angle sequences, and therefore there is a unique point $p\in S$ with the lexicographically smallest angle sequence: $p$ is called the \emph{head} of $S$.

The Gathering algorithm uses the fact that all robots agree on where the head of the swarm is, and the robot located at the head is elected the leader. The algorithm makes the leader move clockwise to the next robot, while all other robots wait. As soon as there is a multiplicity point, the closest robot in the clockwise direction moves counterclockwise to the multiplicity point. The process continues until all robots have gathered.

Note that this algorithm also works in ASYNCH and with non-rigid robots (i.e., robots that can be stopped by an adversary before reaching their destination). Indeed, as the leader moves toward the next robot, its angle sequence remains the lexicographically smallest, and so it remains the leader. After a multiplicity point has been created, only one robot is allowed to move at any time, and therefore no other multiplicity points are accidentally formed.

Before delving into our main algorithm, we need definitions and lemmas that will justify some of our design choices.

\paragraph*{Undecided leaders and cognizant leader}
Let us now consider a swarm of robots with visibility range $\vartheta=\pi$ forming a rotationally asymmetric set $S$ of $n$ distinct points. We say that the \emph{true leader} of the swarm is the robot located at the head of $S$.

For each robot $r$, we define the \emph{visible configuration} $V(r)$ as the set of robots that are visible to $r$, and the \emph{ghost configuration} as $G(r)=V(r)\cup\{r'\}$, with $r'$ antipodal to $r$. Note that exactly one between $V(r)$ and $G(r)$ is isometric to the ``real'' configuration $S$, and therefore at least one between $V(r)$ and $G(r)$ is rotationally asymmetric.

The \emph{visible head} $v(r)$ is defined as follows: if $V(r)$ is rotationally asymmetric, $v(r)$ is the head of $V(r)$; otherwise, $v(r)$ is the head of $G(r)$. The \emph{ghost head} $g(r)$ is defined similarly: if $G(r)$ is rotationally asymmetric, $g(r)$ is the head of $G(r)$; otherwise, $g(r)$ is the head of $V(r)$. Note that the true leader of the swarm must be either $v(r)$ or $g(r)$.

If $v(r)\neq g(r)$, and either $r=v(r)$ or $r=g(r)$, then $r$ is said to be an \emph{undecided leader}: a robot that is possibly the true leader, but does not know for sure. If $r=v(r)=g(v)$, then $r$ is a \emph{cognizant leader}: a robot that is aware of being the true leader of the swarm.

\begin{proposition}\label{p:possleader}
In a rotationally asymmetric swarm with no multiplicity points, the true leader is either an undecided or a cognizant leader, and no robot other than the true leader can be a cognizant leader.
\end{proposition}
\begin{proof}
Let the swarm form configuration $S$, and let $h$ be the head of $S$, i.e., the true leader. Either $V(h)=S$ or $G(h)=S$, which means that either $v(h)=h$ or $g(h)=h$: therefore, $h$ is an undecided or a cognizant leader.

Conversely, if a robot $r$ is a cognizant leader, then $v(r)=g(r)=r$, implying that $r$ is the head of both $V(r)$ and $G(r)$. But either $V(r)=S$ or $G(r)=S$, and hence $r$ must be the head of $S$, i.e., the true leader.
\end{proof}

\paragraph*{Point-addition lemma}
We will now present some technical results that have important implications on the design of our Gathering algorithm. Although these results logically belong before the description of the algorithm itself, they may be skipped on a first reading, as they may be better appreciated after becoming familiar with the algorithm.

\begin{proposition}\label{p:middle}
Let $S$ be a rotationally asymmetric finite set of points on a circle $C$, let $a,b\in S$, and let $c\in C$, such that $0< cw(a,b)=cw(b,c)<\pi$. If the angle sequence of $a$ truncated at $b$ is equal to the angle sequence of $b$ truncated at $c$, then $b$ is not the head of $S$.
\end{proposition}
\begin{proof}
If $x\in S$ and $y\in C$, we denote by $W(x)$ the angle sequence of $x$, and by $W(x,y)$ the angle sequence of $x$ truncated at $y$. Suppose for a contradiction that $b$ is the head of $S$, and hence $W(b)\prec W(a)$. Observe that $c\in S$, otherwise the hypotheses that $W(a,b)=W(b,c)$ and $b\in S$ would imply that $W(a)\prec W(b)$. As a consequence, the angle sequence $W(c)$ is well defined, and so is $W(c,y)$ for $y\in C$.

Note that the conditions on $a$, $b$, and $c$ imply that they are three distinct points, and that the clockwise arcs $ab$, $bc$, and $ca$ are internally disjoint (refer to \cref{f:p8l4} (left)). So, if we let $X=W(a,b)=W(b,c)$ and $Y=W(c,a)$, we obtain $W(a)=XXY$, $W(b)=XYX$, and $W(c)=YXX$. Since $b$ is the head of $S$, we have $XYX=W(b)\prec W(a)=XXY$ and $XYX=W(b)\prec W(c)=YXX$. From the first inequality we derive $YX\prec XY$, and from the second we derive $XY\prec YX$, a contradiction.
\end{proof}

\begin{figure}%[h!]
\begin{center}
\includegraphics[scale=1]{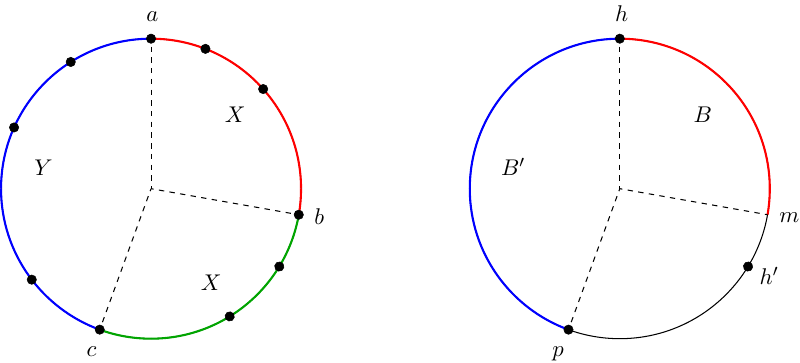}
\end{center}
\caption{Illustrations of \cref{p:middle} (left) and \cref{l:addition} (right).}
\label{f:p8l4}
\end{figure}

\begin{lemma}\label{l:addition}
Let $S$ be a finite non-empty set of points on a circle $C$, and let $S'=S\cup \{p\}$, where $p\in C\setminus S$. Assume that $S$ and $S'$ are rotationally asymmetric, and let $h\in S$ be the head of $S$ and $h'\in S'$ be the head of $S'$. Then, either $h=h'$ or $cw(h, p)>2\cdot cw(h',p)$.
\end{lemma}
\begin{proof}
Let $m\in C$ be the midpoint of the clockwise arc $hp$, i.e., $cw(h,p)=2\cdot cw(h,m)$. Let $B=\{x\in C \mid 0<cw(h,x)\leq cw(h,m)\}$ and $B'=\{x\in C \mid 0<cw(p,x)< cw(p,h)\}$. We have to prove that $h'\notin B\cup B'$, as in \cref{f:p8l4} (right). For the sake of contradiction, assume the contrary. Then, since $h,p\notin B\cup B'$, we have $h\neq h'\neq p$, as well as $p\neq h$, because $p\notin S$. It follows that $h'\in S$. Let $W(x)$, with $x\in S$, be the angle sequence of $x$ with respect to $S$, and let $W'(x)$, with $x\in S'$, be the angle sequence of $x$ with respect to $S'$. If $y\in C$, we define $W(x,y)$ (respectively, $W'(x,y)$) as the angle sequence $W(x)$ (respectively, $W'(x)$) truncated at $y$.

Suppose first that $h'\in B$, and note that $0< cw(h,h')\leq cw(h,m)=cw(h,p)/2<\pi$. Let $h''\in C$ be such that $cw(h,h')=cw(h',h'')$, and observe that $h''$ lies on the clockwise arc $h'p$ (possibly, $h''=p$). Since $W(h)\prec W(h')$, by \cref{p:angle} we have $W(h,h')\preceq W(h',h'')$; similarly, since $W'(h')\prec W'(h)$, we have $W(h',h'')=W'(h',h'')\preceq W'(h,h')=W(h,h')$. Thus, $W'(h,h')=W'(h',h'')$. Observe that the hypotheses of \cref{p:middle} are satisfied by the set $S'$, with $a=h$, $b=h'$, and $c=h''$. We conclude that $h'$ is not the head of $S'$, a contradiction.

Suppose now that $h'\in B'$. Let $p'\in C$ be such that $cw(h,p)=cw(h',p')$, and observe that $p$ does not lie on the clockwise arc $h'p'$ (in particular, $p\neq p'$). Since $W(h)\prec W(h')$, by \cref{p:angle} we have $W(h,p)\preceq W(h',p')$; similarly, since $W'(h')\prec W'(h)$, we have $W(h',p')=W'(h',p')\preceq W'(h,p)=W(h,p)$. Thus, $W(h,p)=W(h',p')=W'(h,p)=W'(h',p')$. Since $p\notin S$, we must also have $p'\notin S$, otherwise we would have $W(h')\prec W(h)$. Similarly, since $p\in S'$, we must also have $p'\in S'$, otherwise we would have $W'(h)\prec W'(h')$. So, $p'\in S'\setminus S$, implying that $p'=p$, a contradiction.
\end{proof}
\begin{corollary}\label{c:cognizant}
In a rotationally asymmetric swarm with no multiplicity points, a robot $r$ is a cognizant leader if and only if $r=g(r)$.
\end{corollary}
\begin{proof}
If $r$ is a cognizant leader, then obviously $r=g(r)$. For the converse, assume that $r=g(r)$. If $V(r)$ or $G(r)$ is rotationally symmetric, then $v(r)=g(r)$ by definition, and therefore $r$ is a cognizant leader. If $V(r)$ and $G(r)$ are rotationally asymmetric, we can apply \cref{l:addition} with $S=V(r)$ and $S'=G(r)=V(r)\cup\{r'\}$ (and hence $h=v(r)$ and $h'=g(r)$), obtaining either $v(r)=g(r)$ or $cw(v(r), r')>2\cdot cw(g(r),r')$. However, since $r=g(r)$, the second condition is impossible, because it implies that $cw(v(r), r')>2\cdot cw(r,r')=2\pi$. It follows that $v(r)=g(r)=r$, and so $r$ is a cognizant leader.
\end{proof}

\paragraph*{Gathering algorithm}

\begin{lstlisting}[caption={Gathering algorithm for $\vartheta=\pi$\label{l:algorithm}},label=list:8-6,captionpos=t,float,abovecaptionskip=-\medskipamount,mathescape=true,backgroundcolor = \color{Azure1}, basicstyle=\small]
The algorithm is executed by a generic robot $r$.
Input: $V(r)$, the set of points occupied by robots visible to $r$
       (expressed in $r$'s coordinate system), with weak multiplicity.
Output: a destination point for $r$.

Let $s\in V(r)$ be such that $cw(r,s)>0$ is minimum, if it exists
($s$ is the visible robot closest to $r$ in the clockwise direction).

Let $V'(r)$ be the set of all the points in $V(r)$ (without multiplicity)
plus their antipodal points.

Let $\delta$ be the smallest $cw(a,b)>0$ with $a,b\in V'(r)$.

 $\bf{1.}$ If $r$ sees some multiplicity points, then:
   $\bf{1.a.}$ If $r$ sees a unique multiplicity point, then $r$ moves to it.
   $\bf{1.b.}$ If $r$ sees two multiplicity points $a$ and $b$ with $cw(a,b)>cw(b,a)$,
       then $r$ moves to $a$.
 $\bf{2.}$ Else, if $r$ sees no other robots, then $r$ moves clockwise by $\pi/2$.
 $\bf{3.}$ Else, if $r=g(r)$, then $r$ moves to $s$.
 $\bf{4.}$ Else, if $r=v(r)$, then:
   $\bf{4.a.}$ If $g(r)$ is antipodal to $r$, and $s$ has an antipodal robot,
       then $r$ moves clockwise by $cw(r,s)+\delta/3$.
   $\bf{4.b.}$ Else, if there is a point $m\in V'(r)$ such that $cw(r,m)=cw(r,s)/2$,
       then $r$ moves clockwise by $cw(r,s)/2+\delta/7$.
   $\bf{4.c.}$ Else, $r$ moves clockwise by $cw(r,s)/2$.
 $\bf{5.}$ Else, $r$ does not move.
\end{lstlisting}
Our Gathering algorithm for $\vartheta=\pi$ is illustrated in \cref{l:algorithm,f:rules3-4c,f:rules4b-4a}. A robot $r$ executing the algorithm first checks if the current configuration falls under some special cases (which will be discussed later), and then it attempts to determine the true leader of the swarm. By \cref{c:cognizant}, checking if $r=g(r)$ is equivalent to checking if $r$ is a cognizant leader. In this case, by \cref{p:possleader}, $r$ is the true leader, and hence it behaves like in the full-visibility algorithm: it moves clockwise to the next robot, $s$ (rule~3).

If $r$ is not a cognizant leader, it checks if it is at least an undecided leader: $r=v(r)$. In this case, $r$ cannot commit itself to moving to $s$, because several robots may be undecided leaders, and this would create more than one multiplicity point. Instead, $r$ attempts to ``strengthen its leadership'' by moving halfway toward $s$ (rule~4.c): this ensures that, in the next turn, $r$ will have a lexicographically smaller angle sequence than it currently has (unless, of course, $s$ moves as well).

\begin{figure}%[h!]
\begin{center}
\includegraphics[scale=1]{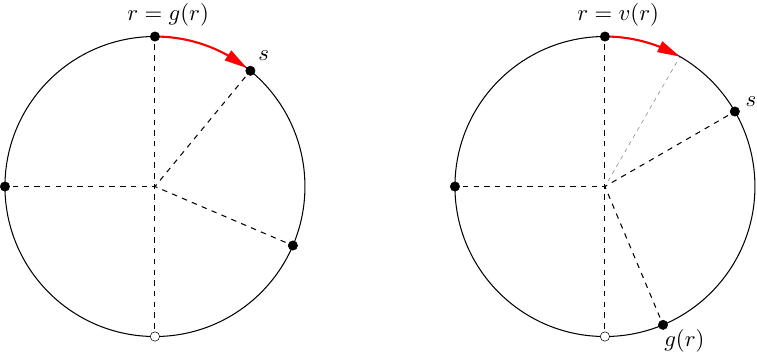}
\end{center}
\caption{Examples of rule~3 (left) and rule~4.c (right) of the Gathering algorithm in \cref{l:algorithm}. Black dots indicate robots that are visible to $r$. A white dot indicates the point antipodal to $r$, which may or may not be occupied by a robot.}
\label{f:rules3-4c}
\end{figure}

Another goal of $r$ is to be able to see the entire swarm in the next turns. Therefore, if the midpoint of $r$ and $s$ happens to be antipodal to some robot $q$, then $r$ moves a bit further past the midpoint (rule~4.b). This way, $r$ will be sure to see $q$ in the next turn (unless, of course, $q$ moves as well).

An exception to the above is when $g(r)$ is antipodal to $r$, and $s$ has an antipodal robot $s'$. In this situation, if $r$ had an antipodal robot $r'$, then $r'$ would be the true leader, which would then either form a multiplicity point with $s'$ (if $r'$ is activated) or would become visible to $r$ (if $r$ is activated but not $r'$). However, if $r'$ does not exist, then $r$ is the true leader, but $r$ may never find out: it may keep approaching $s$ without ever reaching it, and there may always be a ghost head antipodal to $r$. For this reason, if $r$ detects this configuration, it moves slightly past $s$ (rule~4.a): this way, $s$ will be the new leader, and it will not be in the same undesirable configuration, because $r$ will not have an antipodal robot.

\begin{figure}%[h!]
\begin{center}
\includegraphics[scale=1]{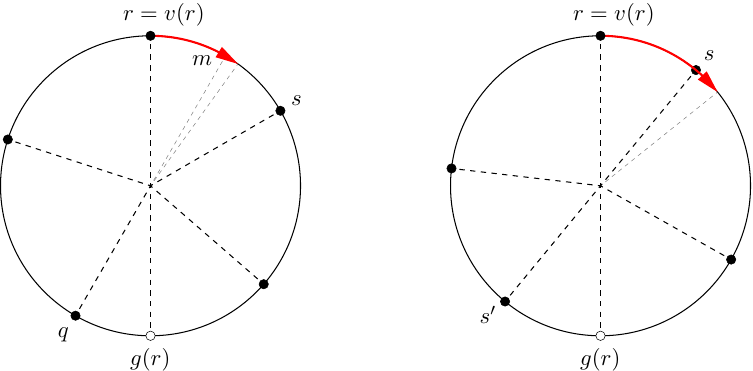}
\end{center}
\caption{Examples of rule~4.b (left) and rule~4.a (right) of the Gathering algorithm in \cref{l:algorithm}}
\label{f:rules4b-4a}
\end{figure}

Note that, when describing rule~4.a and rule~4.b, we mentioned some undefined ``small'' distances. According to \cref{l:algorithm}, these are respectively $\delta/3$ and $\delta/7$. In turn, $\delta$ is defined as the smallest angular distance between two points in $V'(r)$, where $V'(r)$ is the set of all the points in $V(r)$ plus their antipodal points. It is easy to see that all robots that are activated at the same time compute isometric sets $V'$, and therefore they implicitly agree on the value of $\delta$. The reason why the specific values $\delta/3$ and $\delta/7$ have been chosen will become apparent in the proof of correctness of the algorithm.

Finally, let us discuss the special cases, all of which arise when some multiplicity points have been created (due to a cognizant leader moving to some other robot). If only one multiplicity point is visible to $r$, then $r$ simply moves to it, as in the full-visibility algorithm (rule~1.a). In some exceptional circumstances, two multiplicity points $a$ and $b$ may be created, but we will prove that $a$ and $b$ will not be antipodal to each other, and there will never be a third multiplicity point. In this case, there is an implicit order between $a$ and $b$ on which all robots agree, and so they will all move to the same multiplicity point, say $a$ (rule~1.b). The last special case is when all robots have gathered in a point, except a single robot $r$ located in the antipodal point. $r$ detects this situation because it sees no robots other than itself (and its current location is not a multiplicity point). So, $r$ just moves to another visible point, say, the one forming a clockwise angle of $\pi/2$ with $r$ (rule~2). This ensures that $r$ will see the multiplicity point on its next turn.

\paragraph*{Correctness}
In the following, we will assume that all robots in a swarm of size $n>1$ execute the algorithm in \cref{l:algorithm} starting from a rotationally asymmetric initial configuration with no multiplicity points. We will prove that, no matter how the adversarial semi-synchronous scheduler activates them, all robots will eventually gather in a point and no longer move.

We say that, in a given configuration $S$, a robot $r$ \emph{is able to apply rule~$j$} if, assuming that $r$ is activated when the swarm forms $S$, $r$ executes rule~$j$ (and no other rule).

\begin{lemma}\label{l:algomove}
Assume that the swarm forms a rotationally asymmetric configuration with no multiplicity points, and let $\ell$ be the true leader. Then:
\begin{enumerate}
\item No robot is able to apply rule~1 or rule~2.
\item At most one robot is able to apply rule~3: the true leader $\ell$.
\item At most one robot is able to apply rule~4.a: either the true leader $\ell$ (if there is no robot antipodal to $\ell$) or the robot antipodal to $\ell$.
\item A robot $r\neq \ell$ is able to apply rule~4 only if $\pi/2 < cw(r,\ell) \leq \pi$, and only if there is a robot antipodal to $r$.
\end{enumerate}
\end{lemma}
\begin{proof}
To prove the first statement, observe that rule~1 and rule~2 only apply when there are multiplicity points, with one exception: when the swarm consists of exactly two robots, which are antipodal to each other. However, such a configuration is rotationally symmetric.

For the second statement, note that a robot $r$ is able to apply rule~3 only if $r=g(r)$, which, by \cref{c:cognizant}, implies that $r$ is a cognizant leader. Therefore, by \cref{p:possleader}, $r=\ell$.

To prove the third statement, assume that robot $r$ is able to apply rule~4.a, which means that $r=v(r)$ and $g(r)$ is antipodal to $r$. Recall that either $\ell=v(r)$ or $\ell=g(r)$. Specifically, if there is a robot in $g(r)$, then this robot is $\ell$, and therefore $r$ is antipodal to $\ell$. Otherwise there is no robot in $g(r)$, and so $\ell=v(r)=r$.

Let us now prove the fourth statement, which is illustrated in \cref{f:l5l6} (left). Assume that a robot $r\neq \ell$ is able to apply rule~4, which means that $r=v(r)$, and therefore $\ell=g(r)$. Note that, since the true leader is $g(r)$, the swarm's configuration is $G(r)$, and so there must be a robot $r'$ antipodal to $r$. Also, both $V(r)$ and $G(r)=V(r)\cup \{r'\}$ must be rotationally asymmetric, otherwise we would have $v(r)=g(r)$, by definition of visible and ghost head. Thus, we can apply \cref{l:addition} with $S=V(r)$ and $S'=G(r)$ (hence $p=r'$, $h=v(r)=r$, and $h'=g(r)=\ell$) to conclude that either $r=\ell$ (which is not the case) or $cw(r,r')>2\cdot cw(\ell, r')$. Since $cw(r,r')=\pi$, we have $cw(\ell, r')<\pi/2$. This inequality implies that $\ell$ lies on the clockwise arc $rr'$, and therefore $cw(r,\ell)+cw(\ell,r')=cw(r,r')=\pi$, or equivalently, $cw(\ell,r')=\pi-cw(r,\ell)$. By plugging this into our last inequality, we obtain $\pi/2 < cw(r,\ell)$. Also, from the obvious fact that $cw(\ell,r')\geq 0$, we deduce that $cw(r,\ell)\leq \pi$.
\end{proof}

In light of \cref{l:algomove}, if the swarm forms a rotationally asymmetric configuration with no multiplicity points, we say that a robot \emph{is able to move} if it is able to apply rule~3 or rule~4: indeed, these are the only rules that result in a non-null movement.

\begin{corollary}\label{c:algomove}
In any rotationally asymmetric configuration with no multiplicity points, a robot is able to move if and only if it is an undecided or a cognizant leader.
\end{corollary}
\begin{proof}
By the first statement of \cref{l:algomove}, no robot is able to apply rule~1 or rule~2. Hence, a robot $r$ is able to apply rule~3 if and only if $r=g(r)$, and it is able to apply rule~4 if and only if $r=v(r)\neq g(r)$. By \cref{c:cognizant}, $r=g(r)$ if and only if $r$ is a cognizant leader, and so $r=v(r)\neq g(r)$ if and only if $r$ is an undecided leader.
\end{proof}

\begin{figure}%[h!]
\begin{center}
\includegraphics[scale=1]{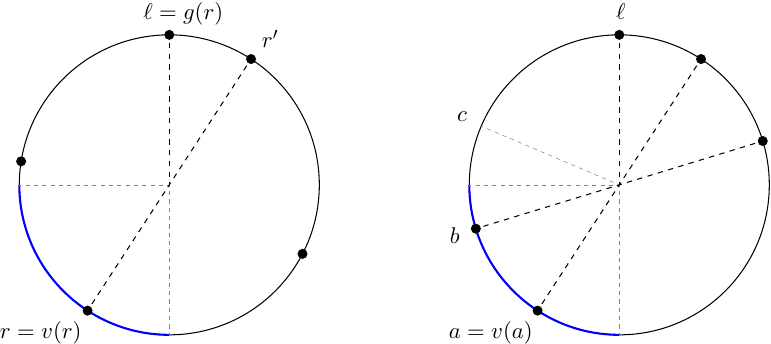}
\end{center}
\caption{Illustrations of the fourth statement of \cref{l:algomove} (left) and \cref{l:algotwo} (right).}
\label{f:l5l6}
\end{figure}

\begin{lemma}\label{l:algotwo}
In any rotationally asymmetric configuration with no multiplicity points, at most two robots are able to move, and the true leader is always able to move.
\end{lemma}
\begin{proof}
By \cref{p:possleader}, the true leader is either an undecided or a cognizant leader, and therefore, by \cref{c:algomove}, it is able to move.

Suppose for a contradiction that, in some rotationally asymmetric configuration with no multiplicity points, at least three distinct robots are able to move. One of these robots must be the true leader $\ell$, hence there are two robots $a$ and $b$, distinct from $\ell$, that are able to move. Due to the second statement of \cref{l:algomove}, only $\ell$ is able to apply rule~3, and therefore $a$ and $b$ are able to apply rule~4. By the fourth statement of \cref{l:algomove}, we have $\pi/2<cw(a,\ell)\leq \pi$ and $\pi/2<cw(b,\ell)\leq \pi$. It follows that the angular distance between $a$ and $b$ must be less than $\pi/2$, and so, without loss of generality, we may assume that $0<cw(a,b)<\pi/2$. In particular, $a$ and $b$ see each other, i.e., $a\in V(b)$ and $b\in V(a)$. Refer to \cref{f:l5l6} (right) for an example.

For any $x\in V(a)$ (respectively, $x\in V(b)$) and any point $y$ on the circle, we let $W(x,y)$ (respectively, $W'(x,y)$) be the angle sequence of $x$, with respect to $V(a)$ (respectively, $V(b)$), truncated at $y$. Let $c$ be the point on the circle such that $cw(a,b)=cw(b,c)$. Since $a$ is able to apply rule~4, we have $a=v(a)\neq g(a)$, i.e., $a$ is the head of $V(a)$. In particular, $W(a)\prec W(b)$ and, by \cref{p:angle}, $W(a,b)\preceq W(b,c)$. By a similar chain of deductions, since $b$ is able to apply rule~4, we have $W'(b)\prec W'(a)$ and $W'(b,c)\preceq W'(a,b)$. Recall that $cw(a,b)<\pi/2$, and therefore $cw(a,c)<\pi$, which implies that the antipodal points of $a$ and $b$ lie outside of the clockwise arc $ac$. But $V(a)$ and $V(b)$ only differ by the antipodal points of $a$ and $b$, and so $W(a,b)=W'(a,b)$ and $W(b,c)=W'(b,c)$. Thus, from $W(a,b)\preceq W(b,c)$ we derive $W'(a,b)\preceq W'(b,c)$, which, together with $W'(b,c)\preceq W'(a,b)$, yields $W'(a,b)=W'(b,c)$. Since $V(b)$ is rotationally asymmetric (or else we would have $v(b)=g(b)$) and $a,b\in V(b)$, it follows that $V(b)$ satisfies the hypotheses of \cref{p:middle}. We conclude that $b$ is not the head of $V(b)$, a contradiction.
\end{proof}

\begin{lemma}\label{l:algomult}
If the swarm has a unique multiplicity point, then all robots eventually gather and no longer move.
\end{lemma}
\begin{proof}
Let $a$ be the unique multiplicity point. Note that all robots are able to see $a$ and will move to it by rule~1.a when activated, except perhaps a single robot $r$, antipodal to $a$. If all robots except $r$ move to $a$, while $r$ is not activated or only applies rule~5, then $r$ will eventually apply rule~2, moving to a different point. After that, $r$ will be able to see $a$, and will move to it by rule~1.a.

Otherwise, if $r$ makes a move before all other robots are in $a$, it does so by applying rule~3 or rule~4. If $r$ applies rule~4, then it will not create a new multiplicity point, all robots (including $r$) will thus see the unique multiplicity point $a$, and will gather in it.

If $r$ applies rule~3, it moves to another robot $s$'s location. If $s$ is activated at the same turn, then $s$ moves to $a$, no new multiplicity point is created, and all robots will eventually gather in $a$. If $s$ is not activated at the same turn as $r$, then $r$ will move to $s$ and create a new multiplicity point $b$. Since rule~3 causes $r$ to move clockwise by an angle smaller than $\pi$, we have $cw(a,b)>\pi>cw(b,a)$. Now all robots see $a$ and perhaps also $b$: either way, they will apply rule~1.a or rule~1.b, both of which result in a move to $a$.
\end{proof}

\begin{lemma}\label{l:algocollision}
Assume that the swarm forms a rotationally asymmetric configuration with no multiplicity points, and let $r$ be a robot that is able to move. Then:
\begin{enumerate}
\item If $r$ is activated and executes rule~4, then it moves by an angular distance strictly smaller than $\pi$, and its destination point is not in $V'(r)$, as defined in \cref{l:algorithm}.
\item If $r$ is activated and executes rule~3, and $r'$ is a robot that is able to move but is not activated, then $r$ does not move to $r'$.
\item If both $r$ and $r'$ are activated and execute rule~4.b or rule~4.c, then $r$ and $r'$ are not antipodal to each other, and their destination points are not antipodal to each other.
\end{enumerate}
\end{lemma}
\begin{proof}
Since $r$ is able to move, by \cref{c:algomove}, it is either an undecided or a cognizant leader, implying that $r$ is the head of $V(r)$ or of $G(r)$. Also, $r$ sees at least one other robot $q\neq r$, or else it would be able to apply rule~2. If there is no robot $q'$ such that $0< cw(r,q')<\pi$, then the first element in the angle sequence of $r$ with respect to both $V(r)$ and $G(r)$ is at least $\pi$. However, this also means that $0<cw(q,r)<\pi$, and therefore $q$ has a lexicographically smaller angle sequence than $r$ with respect to both $V(r)$ and $G(r)$, contradicting the fact that $r$ is the head of one of them. It follows that there must be some robot $q$ with $0<cw(r,q)<\pi$. In particular, if $s$ is defined as in \cref{l:algorithm}, we have $0<cw(r,s)<\pi$.

Let us prove the first statement: assume that $r$ executes rule~4, and let $p$ be its destination point. According to \cref{l:algorithm}, if $r$ executes rule~4.c, then $p\notin V'(r)$ and $cw(r,p)=cw(r,s)/2$, implying that $cw(r,p)<\pi/2<\pi$. If $r$ executes rule~4.b, then $cw(r,p)=cw(r,s)/2+\delta/7$. By definition of $\delta$, we have $\delta\leq cw(r,s)$, and hence $cw(r,p)\leq cw(r,s)/2+cw(r,s)/7< cw(r,s)<\pi$. This rule is executed only if the midpoint of $r$ and $s$ is $m\in V'(r)$. Since $0<cw(m,p)=\delta/7<\delta$, then $p\notin V'(r)$, by definition of $\delta$. Finally, if $r$ executes rule~4.a, we have $cw(r,p)=cw(r,s)+\delta/3$. Let $a\in V'(r)$ be the antipodal point of $r$, and note that, by definition of $\delta$, we have $\delta\leq cw(s,a)$. It follows that $cw(r,p)=cw(r,s)+\delta/3<cw(r,s)+\delta\leq cw(r,s)+cw(s,a)=cw(r,a)=\pi$ (note that we can write $cw(r,s)+cw(s,a)=cw(r,a)$ because $s$ is on the clockwise arc $ra$). Moreover, since $0<cw(s,p)=\delta/3<\delta$ and $s\in V'(r)$, then we have $p\notin V'(r)$, by definition of $\delta$.

Let us prove the second statement, which is illustrated in \cref{f:l8} (left). By repeating the above reasoning with $r'$, we infer that its next robot in the clockwise direction, $s'$, satisfies $0<cw(r',s')<\pi$. If $r$ executes rule~3, then $r$ is the true leader, and $cw(r',r)\leq \pi$, by \cref{l:algomove}. Thus, $cw(r,r')\geq \pi$, implying that $r'$ is distinct from the destination point $s$ of $r$, because $cw(r,s)<\pi$. We conclude that $r$ does not move to $r'$.

Let us now prove the third statement. Since both $r$ and $r'$ are able to move, due to \cref{l:algotwo}, one of them must be the true leader, say $r$. By the fourth statement of \cref{l:algomove}, $cw(r',r)\leq \pi$, and $r'$ has an antipodal robot. Assume for a contradiction that $r$ is antipodal to $r'$. This means that $G(r)$ coincides with the current configuration, and hence $r$ is the head of $G(r)$, i.e., $r=g(r)$. So, $r$ is able to apply rule~3, which contradicts the fact that it executes rule~4.b or rule~4.c. We have established that $r$ has no antipodal robot, and the robot antipodal to $r'$ is $a\neq r$ (refer to \cref{f:l8} (right)). Also, because $cw(r',r)\leq \pi$, the robot $r$ is located in the interior of the clockwise arc $r'a$. If $s$ is defined as in \cref{l:algorithm}, and $p$ is the destination point of $r$, according to rule~4.b and rule~4.c, we have $cw(r,p)<cw(r,s)\leq cw(r,a)$. So, the angular distance between $r$ and $p$ is smaller than $cw(r,a)$. On the other hand, if $p'$ is the destination point of $r'$, and $s'$ is defined as above, we have $0<cw(r',p')<cw(r',s')\leq cw(r',r)$. Hence, the angular distance between $r$ and $p'$ is smaller than $cw(r',r)$. By the triangle inequality, the angular distance between $p$ and $p'$ is smaller than $cw(r',r)+cw(r,a)=cw(r',a)=\pi$, implying that $p$ and $p'$ are not antipodal to each other.
\end{proof}

\begin{figure}%[h!]
\begin{center}
\includegraphics[scale=1]{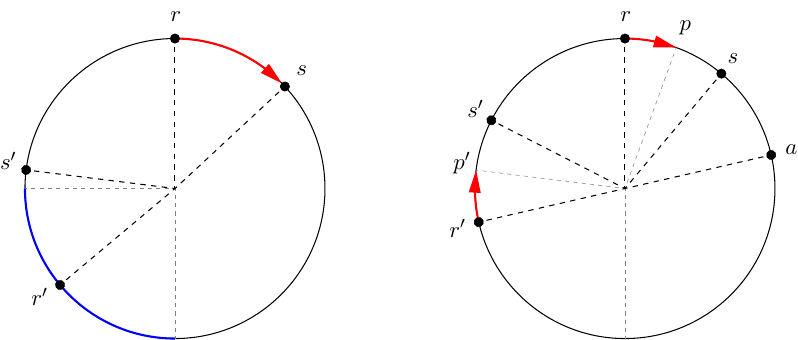}
\end{center}
\caption{Illustrations of the second statement (left) and the third statement (right) of \cref{l:algocollision}.}
\label{f:l8}
\end{figure}

\begin{lemma}\label{l:algorule3}
If the swarm forms a rotationally asymmetric configuration with no multiplicity points and an active robot executes rule~3, then all robots eventually gather in a point and no longer move.
\end{lemma}
\begin{proof}
By \cref{l:algomove}, the robot that executes rule~3 must be the true leader $\ell$, and any robot $r\neq \ell$ that moves at the same time as $\ell$ executes rule~4. Moreover, by \cref{l:algotwo}, if such a robot $r$ exists, it is unique. According to rule~3, the destination point of $\ell$ is the location of another robot $s$. So, if only $\ell$ moves, a unique multiplicity point is created on $s$, and \cref{l:algomult} allows us to conclude the proof.

Otherwise, assume that $r$ exists and moves at the same time as $\ell$. By the second statement of \cref{l:algocollision}, $s\neq r$, and so $\ell$ still creates a multiplicity point on $s$. On the other hand, by the first statement of \cref{l:algocollision}, $r$ does not create a multiplicity point distinct from $s$. Again, there is a unique multiplicity point, and we can apply \cref{l:algomult}.
\end{proof}

\begin{lemma}\label{l:algorule4a}
If the swarm forms a rotationally asymmetric configuration with no multiplicity points and an active robot executes rule~4.a, then all robots eventually gather in a point and no longer move.
\end{lemma}
\begin{proof}
Let $t$ be the time at which some robot executes rule~4.a. By the third statement of \cref{l:algomove}, such a robot must be either the true leader $\ell$ or the robot $\ell'$ antipodal to $\ell$. If it is $\ell'$, we know by \cref{l:algotwo} that only $\ell$ and $\ell'$ are able to move at time $t$. However, in this case, $\ell$ is necessarily able to apply rule~3 at time $t$: indeed, since $\ell$ has an antipodal robot, the current configuration coincides with $G(\ell)$, and therefore $\ell$ is the head of $G(\ell)$, i.e., $\ell=g(\ell)$. Thus, if $\ell$ is activated at time $t$, it executes rule~3, and we conclude the proof by \cref{l:algorule3}. So, if $\ell'$ executes rule~4.a at time $t$, we may assume that no other robot moves at the same time. On the other hand, if $\ell$ executes rule~4.a at time $t$, by \cref{l:algomove,l:algotwo} there is at most one other robot $a$ that is able to move at time $t$. Such a robot $a$ is not antipodal to $\ell$ (otherwise $\ell$ would be able to execute rule~3, not rule~4.a), it is not able to execute rule~4.a (only $\ell$ or the robot antipodal to $\ell$ could), it has an antipodal robot $a'$, and $\pi/2<cw(a,\ell)\leq \pi$.

In any case, we have the following situation at time $t$ (illustrated in \cref{f:l10}): a robot $r$ that executes rule~4.a (it does not matter if $r=\ell$ or $r=\ell'$), and possibly another robot $a$ that executes rule~4.b or rule~4.c, with an antipodal robot $a'\neq r$, and such that $\pi/2<cw(a,r)<\pi$ (since $a$ is not antipodal to $r$). Let $s$ be the robot next to $r$ in the clockwise direction, and let $\alpha=cw(r,s)$. According to \cref{l:algorithm}, $r=v(r)$, there is a robot $s'$ antipodal to $s$, and $r'=g(r)$, with $r'$ antipodal to $r$. Note that there is no robot $x\neq r'$ such that $cw(x,r')<\alpha$, or else $r'\neq g(r)$. We conclude that, whether $r'$ is occupied by a robot or not, $\alpha$ is the minimum angular distance between any two robots at time $t$. Therefore, if $q$ is the robot next to $a$ in the clockwise direction, and $\beta=cw(a,q)$, we have $cw(a,r)\geq cw(a,q)=\beta\geq \alpha$. On the other hand, we have $cw(r,a')\geq cw(r,s)$ by definition of $s$. Let $\delta$ be computed as in \cref{l:algorithm}, and recall that all robots that are active at the same time agree on the same $\delta$. Note that $\delta\leq \alpha$ and $\delta\leq \beta$. We distinguish three cases: (i)~$a=s'$, (ii)~$a\neq s'$, and (iii)~$a$ does not exist.

(i)~Let us first assume that $a=s'$ at time $t$ (refer to \cref{f:l10} (left)). Let us examine the configuration at time $t+1$ (hence, where referring to $r$, we mean the position of $r$ at time $t+1$, etc.). According to \cref{l:algorithm}, $cw(s,r)=\delta/3$ (rule~4.a), and either $cw(a,q)=\beta/2-\delta/7\geq \delta/2-\delta/7>\delta/3$ (rule~4.b) or $cw(a,q)=\beta/2\geq \delta/2$ (rule~4.c). The only exceptional case is $q=r$, which implies that $cw(a,q)>\alpha\geq \delta$. In all cases, we have $cw(a,q)>\delta/3=cw(s,r)$. Observe that $s$ no longer has an antipodal robot, and $a$ has moved by at least $\beta/2\geq \delta/2>\delta/3$. It follows that $a$ is not antipodal to $r$, and therefore $r$ has no antipodal robot, either. Moreover, $cw(a,q)>\delta/3$ implies that the angular distance between $s$ and $r$ is strictly smaller than any other, and in particular the configuration at time $t+1$ is still rotationally asymmetric, and $s$ is the new true leader. Also, all robots see both $s$ and $r$ (because $s$ and $r$ have no antipodal robots). So, the only robot $x$ such that $x=v(x)$ is $x=s$, and therefore no robot other than $s$ is able to move at time $t+1$ (recall that, due to the second statement of \cref{l:algomove}, if a robot other than the true leader is able to move, it must be able to apply rule~4). Moreover, the angular distance between the point antipodal to $s$ and the next robot in the clockwise direction, i.e., $a$, is at least $\beta/2>\delta/3$, which implies that $s=g(s)$. It follows that $s$ is able to apply rule~3 at time $t+1$. So, the configuration will not change as long as $s$ is not activated, and eventually the semi-synchronous scheduler will activate $s$. At that point, $s$ will execute rule~3, and \cref{l:algorule3} applies.

(ii)~Let us now assume that $a\neq s'$ at time $t$ (refer to \cref{f:l10} (right)). From $cw(a,r)>\pi/2$, we get $cw(r,a')<\pi/2$, and so $\alpha\leq \pi/2$ (recall that $\alpha$ is the minimum angular distance between any two robots at time $t$). Also, since $a=v(a)$ at time $t$ (or else $a$ would not be able to apply rule~4), and since $a$ sees both $r$ and $s$, we have that $cw(a,q)=\alpha$ (or else $a$ would not be the head of $V(a)$). Hence, $cw(a,q)=\alpha\leq \pi/2<cw(a,r)$, implying that $q\neq r$, and therefore $q$ does not move at time $t$. Le us analyze the configuration at time $t+1$. As in case~(i), once again we have $cw(s,r)=\delta/3$, and we can prove that the angular distance between $s$ and $r$ is strictly smaller than any other. In particular, the configuration at time $t+1$ is still rotationally asymmetric, and $s$ is the new true leader. Also, no robot other than $a$ can be antipodal to $r$, by definition of $V'$ and $\delta$. But $a$ cannot be antipodal to $r$ either, because the point $r'$ antipodal to $r$ satisfies $cw(s',r')=\delta/3$, while $cw(s',a)>\alpha/2\geq \delta/2$. So, $r$ has no antipodal robots at time $t+1$. However, unlike in case~(i), $s$ still has an antipodal robot $s'$. Note that $s$ is able to apply rule~3, because $cw(s',a)>cw(s,r)$, and so $s=g(s)$. All robots other that $s'$ see both $s$ and $r$, and are therefore unable to move at time $t+1$. If $s'$ is unable to move or is not activated before $s$, then $s$ executes rule~3, and \cref{l:algorule3} applies.

So, we may assume that $s'$ is the only robot to move, say, at time $t'\geq t+1$, when the configuration is still the same as at time $t+1$. By the second statement of \cref{l:algomove}, $s'$ must execute rule~4, and so $s'=v(s')$ at time $t'$. Let us prove that $s'$ cannot apply rule~4.a. Assume the opposite, and let $u$ be the robot next to $s'$ in the clockwise direction. According to \cref{l:algorithm}, $u$ must have an antipodal robot $u'$, and $u'\neq r$, because $r$ does not have an antipodal robot. Also, since there are no robots between $s$ and $r$, we have $cw(r,u')<cw(s,u')=cw(s',u)$. It follows that $s'$ cannot be the head of $V(s')$, because it sees both $r$ and $u'$, and so $r$ has a lexicographically smaller angle sequence than $s'$ with respect to $V(s')$. So, $s'$ must execute either rule~4.b or rule~4.c at time $t'$, with $\delta'\leq cw(s,r)=\delta/3<\delta$. Let us consider the configuration at time $t'+1$. This is a situation analogous to case~(i) at time $t+1$: $r$ has executed rule~4.a with some $\delta$, and $s'$ has executed rule~4.b or rule~4.c with some $\delta'<\delta$. At this point, our previous argument repeats almost verbatim, with $s'$ instead of $a$ and $u$ instead of $q$: indeed, the only difference is that the inequality $\delta/2-\delta/7>\delta/3$ becomes $\delta/2-\delta'/7>\delta/2-\delta/7>\delta/3$.

(iii)~Finally, let us assume that $a$ does not exist. At time $t+1$, we have $cw(r,s)=\delta/3$, and no robot other than $r$ has moved. This is analogous to the situation of case~(ii) at time $t+1$: indeed, the angular distance between $s$ and $r$ is the smallest, so the configuration is rotationally asymmetric, and $s$ is the true leader. Also, $s$ has an antipodal robot $s'$, and $r$ has no antipodal robot. From here, the proof proceeds verbatim as in case~(ii).
\end{proof}

\begin{figure}%[h!]
\begin{center}
\includegraphics[scale=1]{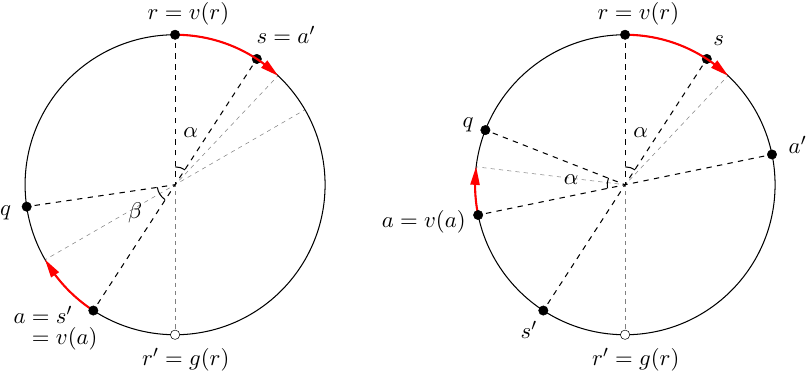}
\end{center}
\caption{Illustrations of \cref{l:algorule4a}.}
\label{f:l10}
\end{figure}

\begin{lemma}\label{l:algonoanti}
Assume that, at time $t$, the swarm forms a rotationally asymmetric configuration with no multiplicity points, and all the robots that move at time $t$ execute rule~4.b or rule~4.c. Then, at time $t+1$, the swarm still forms a rotationally asymmetric configuration with no multiplicity points.
\end{lemma}
\begin{proof}
Due to \cref{l:algotwo}, at most two robots $r$ and $r'$ may move at time $t$. If no robot moves at time $t$, the configuration does not change, and there is nothing to prove. So, we have two cases: (i)~only $r$ moves at time $t$, and (ii)~both $r$ and $r'$ move at time $t$.

(i)~Suppose that only $r$ moves at time $t$. Let $s$ be the next robot in the clockwise direction, and let $p$ be the destination point of $r$. According to \cref{l:algorithm}, $p$ is in the interior of the clockwise arc $rs$, and hence there are no multiplicity points at time $t+1$.

Assume for a contradiction that the configuration at time $t+1$ has a $k$-fold rotational symmetry, with $k>1$. According to \cref{l:algorithm}, $r=v(r)$ at time $t$, which implies that $cw(r,s)$ is the minimum angular distance between any two robots in $V(r)$ at time $t$. If no robot is antipodal to $r$ at time $t$, then no two robots are closer than $r$ and $s$. In this case, at time $t+1$, $r$ and $s$ are strictly closer than any other pair of robots, which implies that $k=1$, a contradiction.

So, we may assume that $r$ has an antipodal robot $b$ at time $t$. Let $a$ (respectively, $c$) be the robot next to $b$ in the counterclockwise (respectively, clockwise) direction. Note that $V(r)$ at time $t$ includes all robots except $b$. So, at time $t+1$, $cw(r,s)$ is strictly smaller than the angular distance between any two robots, except perhaps $a$ and $b$ or $b$ and $c$.

Let us discuss all possibilities for $x=cw(r,s)$, $y=cw(a,b)$, and $z=cw(b,c)$ at time $t+1$ (refer to \cref{f:l11} (left)). If one among $x$, $y$, or $z$ is smaller than the other two, then there is a unique pair of robots that is strictly closer than any other pair, which implies that $k=1$, a contradiction. If $x=y=z$, then there are exactly three pairs of robots that are closest to each other, and therefore $k=3$. This means that the rotational symmetry maps $r$ to $a$, $a$ to $b$, and $b$ to $r$. In particular, we have $cw(r,a)=cw(b,r)$ at time $t+1$, which contradicts the fact that $cw(r,a)<\pi$ and $cw(b,r)>\pi$.

Finally, if two among $x$, $y$, and $z$ are equal and the third one is greater, then there are exactly two pairs of robots that are closest to each other, and therefore $k=2$, i.e., the rotational symmetry exchanges the two pairs. If $x=y<z$, the rotational symmetry should map $s$ to $b$, which is impossible because $s$ and $b$ are not antipodal to each other. If $y=z<x$, the rotational symmetry maps $a$ to $b$, contradicting the fact that they are not antipodal to each other. Similarly, if $z=x<y$, the rotational symmetry maps $r$ to $b$, which is impossible because they are not antipodal to each other at time $t+1$.

(ii)~Assume that both $r$ and $r'$ move at time $t$, as in \cref{f:l11} (right). By \cref{l:algotwo}, one of the two robots, say $r$, is the true leader. So, if $s$ is the robot next to $r$ in the clockwise direction, then $cw(r,s)$ is the minimum angular distance between any two robots at time $t$. By the fourth statement of \cref{l:algomove}, we have $cw(r',r)\leq \pi$. On the other hand, by the third statement of \cref{l:algocollision}, $r$ and $r'$ are not antipodal to each other, and so $cw(r,r')<\pi$. It follows that $s\neq r'$, or else $r'$ would have an angle sequence lexicographically smaller than $r$, contradicting the fact that $r$ is the true leader.

Let $p$ and $p'$ be the destination points of $r$ and $r'$, respectively. According to \cref{l:algorithm}, $p$ is in the interior of the clockwise arc $rs$, and $p'$ is in the interior of the clockwise arc $r's'$, where $s'$ is the robot next to $r'$ in the clockwise direction. Since the two arcs are internally disjoint, we have $p\neq p'$, and therefore there are no multiplicity points at time $t+1$.

Assume for a contradiction that the configuration at time $t+1$ has a $k$-fold rotational symmetry, with $k>1$. Note that, since $s\neq r'$, the robot $s$ does not move at time $t$. Therefore, at time $t+1$, $cw(r,s)$ is strictly smaller than the angular distance between any two robots, except perhaps $r'$ and $s'$. If $cw(r,s)<cw(r's')$ at time $t+1$, then $r$ and $s$ are strictly closer than any other pair of robots, implying that $k=1$, a contradiction. Similarly, if $cw(r,s)>cw(r's')$ at time $t+1$, then $r'$ and $s'$ are strictly closer than any other pair of robots, and once again we have $k=1$. On the other hand, if $cw(r,s)=cw(r's')$ at time $t+1$, then there are exactly two pairs of robots that are closest to each other, and therefore $k=2$, i.e., the rotational symmetry exchanges the two pairs. We conclude that $r$ and $r'$ are antipodal to each other at time $t+1$, or, equivalently, $p$ and $p'$ are antipodal to each other. However, since both $r$ and $r'$ executed rule~4.b or rule~4.c at time $t$, this contradicts the third statement of \cref{l:algocollision}.
\end{proof}

\begin{figure}%[h!]
\begin{center}
\includegraphics[scale=1]{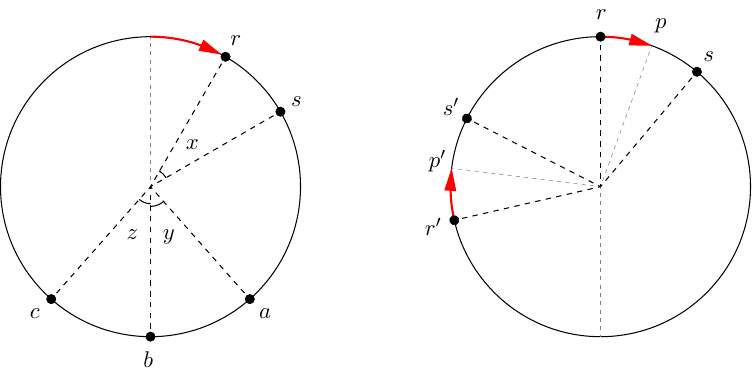}
\end{center}
\caption{Illustrations of \cref{l:algonoanti}.}
\label{f:l11}
\end{figure}

\begin{lemma}\label{l:algoonemove}
Assume that, at time $t$, the swarm forms a rotationally asymmetric configuration with no multiplicity points. If, at all times $t'\geq t$, no robot other than the true leader moves (either because it is unable to move or because it is not activated), then all robots eventually gather in a point and no longer move.
\end{lemma}
\begin{proof}
By \cref{l:algotwo}, the true leader $\ell$ must be able to move at time $t$. Let $s$ be the robot next to $\ell$ in the clockwise direction, and observe that the angular distance between $\ell$ and $s$ at time $t$ is not greater than the angular distance between any two robots in the swarm (or else, $\ell$ would not have the lexicographically smallest angle sequence). As long as $\ell$ is not activated, the configuration will remain the same, and $\ell$ will remain the true leader. The semi-synchronous scheduler will eventually activate $\ell$, say at time $t_0\geq t$, and hence $\ell$ will execute either rule~3 (in which case we conclude the proof by \cref{l:algorule3}) or rule~4.a (in which case we conclude the proof by \cref{l:algorule4a}), or rule~4.b or rule~4.c. In the latter two cases, the angular distance between $\ell$ and $s$ is reduced by at least one half, and so it becomes strictly smaller than the angular distance between any other pair of robots. It follows that the configuration at time $t_0+1$ is still rotationally asymmetric, and $\ell$ is still the true leader.

By inductively applying the same reasoning as above, we argue that either all robots gather in a point and no longer move, or there is an increasing sequence $(t_i)_{i\geq 0}$ such that, for every $i\geq 0$, at time $t_i$ the robot $\ell$ executes rule~4.b or rule~4.c, the angular distance between $\ell$ and $s$ is reduced by at least one half, and $\ell$ remains the true leader. So, if $\alpha_i$ is the value of $cw(\ell,s)$ at time $t_i$, we have $0<\alpha_i\leq \alpha_0/2^i$, that is, $\ell$ keeps approaching $s$ indefinitely.

Let $M$ be a large-enough index such that, at time $t_M$, the angle $\alpha_M$ is strictly smaller than half the angular distance between any robot other than $\ell$ and its next robot clockwise. Let $S$ be the set $G(\ell)$ at time $t_M$. If $S$ is rotationally symmetric, then $\ell$ is a cognizant leader at time $t_M$ (by definition of $v(\ell)$ and $g(\ell)$), and therefore it executes rule~3, due to \cref{c:cognizant}. In this case we can apply \cref{l:algorule3}, so we may assume that $S$ is rotationally asymmetric, and therefore it has a head $h\in S$. Since both $\ell$ and $s$ are in $S$, there must be a point $h'\in S$ such that $cw(h,h')\leq \alpha_M$. So, there are only two possibilities for the location of $h$: either $\ell$ or the point antipodal to $\ell$ (in which case $h'$ must be the point antipodal to $s$). If $\ell$ is the head of $S$, then at time $t_M$ it executes rule~3, and we can apply \cref{l:algorule3}. Otherwise, if the point antipodal to $\ell$ is the head of $S$, it means that $s$ has an antipodal robot. So, $\ell$ executes rule~4.a at time $t_M$, and we conclude the proof by \cref{l:algorule4a}.
\end{proof}

We are now ready to prove the main result of this section.
\begin{theorem}\label{t:algorithm}
If $\vartheta=\pi$, there is a Gathering algorithm under the condition that the swarm initially forms a rotationally asymmetric configuration with no multiplicity points.
\end{theorem}
\begin{proof}
We will prove that the distributed algorithm in \cref{l:algorithm} solves the Gathering problem under the condition that the swarm initially forms a rotationally asymmetric configuration with no multiplicity points. By the first statement of \cref{l:algomove}, whenever the swarm forms a rotationally asymmetric configuration with no multiplicity points, any robot that is activated and moves executes either rule~3 (in which case we conclude the proof by \cref{l:algorule3}), or rule~4.a (in which case we conclude the proof by \cref{l:algorule4a}), or rule~4.b, or rule~4.c. In the latter two cases, by \cref{l:algonoanti}, the resulting configuration is still rotationally asymmetric and with no multiplicity points. By inductively repeating this argument, we may assume, without loss of generality, that the swarm forms a rotationally asymmetric configuration with no multiplicity points at all times, and all robots that are activated and move execute rule~4.b or rule~4.c.

By \cref{l:algotwo}, at a generic time $t$, at least one robot $r$ and at most one other robot $r'\neq r$ are allowed to move. The semi-synchronous scheduler will activate each of them infinitely often, so let $t'\geq t$ be the first time at least one of them is activated. Assume that one robot, say $r'$, is not activated at time $t'$, and therefore $r$ is. Then, $r$ does not have an antipodal robot at time $t'+1$, due to the first statement of \cref{l:algocollision}. Similarly, if both $r$ and $r'$ are activated at time $t'$, none of them has an antipodal robot at time $t'+1$, by the first and third statements of \cref{l:algocollision}.

In summary, if a robot is activated and moves at a generic time $t$, it no longer has antipodal robots at any time after $t$. Since the robots are finitely many, eventually, say after time $t''$, only robots without an antipodal robot will move. However, by the fourth statement of \cref{l:algomove}, a robot that moves must have an antipodal robot, unless it is the current true leader. So, at all times after $t''$, no robot other than the current true leader will move. Therefore, \cref{l:algoonemove} allows us to conclude the proof.
\end{proof}

\section{Conclusions}\label{sec:5}
\paragraph*{Summary}
We gave a deterministic distributed algorithm to solve the Gathering problem on a circle for semi-synchronous rigid mobile robots with chirality (under the necessary condition that the initial configuration is rotationally asymmetric), assuming that each robot can see the entire circle except its antipodal point, i.e., for $\vartheta=\pi$. On the other hand, we proved that no such algorithm exists if $\vartheta\leq\pi/2$, even if the robots know the size of the swarm, and even if the initial configuration is rotationally asymmetric and has a connected visibility graph. We remark that the latter impossibility result is not limited to the Gathering problem, but extends to all Pattern Formation problems where the pattern is rotationally asymmetric.

\paragraph*{Open problems}
By inspecting the proof of correctness of our Gathering algorithm for $\vartheta=\pi$, it is easy to see that its running time is $O(n)$ epochs, where $n$ is the total number of robots, and an ``epoch'' is defined as a minimal timespan where every robot is activated at least once. Indeed, the only step that may take a non-constant number of epochs is discussed in the proof of \cref{t:algorithm}, where a termination condition occurs as soon as a multiplicity point appears or no robots are antipodal to each other. In principle, this could take $O(n)$ epochs in the worst case, but we do not have a concrete example where it actually takes that long. In fact, we believe that our algorithm terminates in $O(1)$ epochs, but a deeper analysis is needed to confirm this conjecture.

Note that the technique of our impossibility proof is ineffective when $\pi/2<\vartheta<\pi$, because it relies on the fact that a robot's snapshot is contained in an open semicircle, which in turn enables combining snapshots as a means to construct configurations where Gathering is impossible. We leave the design of Gathering algorithms or proofs of impossibility for these remaining cases as an open problem.

Other natural open problems concern the relaxation of some of our assumptions concerning the capabilities of the robots and the scheduler. For instance, we can prove that, under an asynchronous scheduler, the algorithm in \cref{l:algorithm} may cause the swarm to partition into two antipodal multiplicity points: a configuration from which Gathering cannot be achieved. Our question is whether a Gathering algorithm exists for an asynchronous scheduler if $\vartheta=\pi$.

Similarly, removing the rigidity condition may cause the formation of two antipodal multiplicity points. For example, a robot $r$ executing rule~4.a may end up stopping on the next robot $s$ instead of slightly past it. If $r$ has an antipodal robot $r'$, then $r'$ will move to the next robot $s'$, creating two antipodal multiplicity points at $s$ and $s'$. It is an open problem whether non-rigid robots can solve the Gathering problem with $\vartheta=\pi$. A related question is whether the chirality assumption can be removed.

Another interesting problem concerns the necessity of (weak) multiplicity detection. It is not clear if the Gathering problem can be solved at all by robots with no multiplicity detection. On the other hand, assuming that the initial configuration has no multiplicity points essentially serves to prevent the occurrence of antipodal multiplicity points. It is an open problem whether this condition can be relaxed somehow.

Conversely, there are open problems arising from robot models with additional capabilities. For example, we may wonder if a fully synchronous scheduler (which activates all robots at all times) allows for simpler or improved Gathering algorithms. Another intriguing question is whether the knowledge of the size of the swarm and stronger multiplicity detection would help in the design of a Gathering algorithm for a large-enough $\vartheta<\pi$.

Finally, the presence of \emph{colored lights} is a feature that is worth investigating. In this model, each robot is able to signal information (usually only a few bits) to all visible robots by means of a colored light. Since our impossibility result does not hold in this model, we wonder if colored lights may allow robots to gather when their visibility range is shorter than $\pi$, or even shorter than $\pi/2$. Colored lights may also be traded with chirality, rigidity, or synchrony.

\bibliography{biblio.bib}

\end{document}